\newcommand{\s}{\mathbf{s}}
\newcommand{\M}{\mathcal{M}}
\newcommand{\gd}[1]{d^{\text{greedy}}_{#1}}
\newcommand{\points}{\mathcal{P}}

\documentclass[acmsmall,nonacm]{acmart}
\usepackage{mathtools}
\DeclarePairedDelimiter\ceil{\lceil}{\rceil}

\usepackage{amsmath}

\DeclareMathOperator*{\argmin}{arg\,min}

\usepackage{thm-restate}
\usepackage{booktabs}
\usepackage{lineno}

\AtBeginDocument{%
  }

\usepackage{cleveref}
\usepackage{graphicx} 
\usepackage
[
]
{subcaption}         
\usepackage{wrapfig} 
\usepackage{tikz}
\usepackage{xspace} 
\newlength{\mytmpa}
\usetikzlibrary{positioning, fit, arrows.meta, angles}
\tikzset{node distance=2cm,
        vertex/.style={circle, draw=black, fill=white, thick, minimum size=10mm, font={\huge}},anchor=center,
        <->/.style={{Latex[length=2mm]}-{Latex[length=2mm]},semithick},
        ->/.style={-{Latex[length=2mm]},semithick},
        <-/.style={{Latex[length=2mm]}-,semithick}}

\begin{document}

\title{Strategic Network Creation for Enabling Greedy Routing}

\author{Julian Berger}
\email{julian.berger@student.hpi.de}
\affiliation{%
  \institution{Hasso Plattner Institute, University of Potsdam, Potsdam}
  \country{Germany}
}
\author{Tobias Friedrich}
\email{friedrich@hpi.de}
\affiliation{%
  \institution{Hasso Plattner Institute, University of Potsdam, Potsdam, \country{Germany}}
}
\author{Pascal Lenzner}
\email{pascal.lenzner@uni-a.de}
\affiliation{%
  \institution{Institute of Computer Science, University of Augsburg, Augsburg, \country{Germany}}
}
\author{Paraskevi Machaira}
\email{paraskevi.machaira@hpi.de}
\affiliation{%
  \institution{Hasso Plattner Institute, University of Potsdam, Potsdam, \country{Germany}}
}
\author{Janosch Ruff}
\email{janosch.ruff@hpi.de}
\affiliation{%
  \institution{Hasso Plattner Institute, University of Potsdam, Potsdam, \country{Germany}}
}

\newtheorem{remark}[theorem]{Remark}

\begin{abstract}

Today we rely on networks that are created and maintained by smart devices. For such networks, there is no governing central authority but instead  the network structure is shaped by the decisions of selfish intelligent agents.    
A key property of such communication networks is that they should be easy to navigate for routing data. For this, a common approach is greedy routing, where every device simply routes data to a neighbor that is closer to the respective destination.    

Networks of intelligent agents can be analyzed via a game-theoretic approach and in the last decades many variants of network creation games have been proposed and analyzed. 
In this paper we present the first game-theoretic network creation model that incorporates greedy routing, i.e., the strategic agents in our model are embedded in some metric space and strive for creating a network among themselves where all-pairs greedy routing is enabled. 
Besides this, the agents optimize their connection quality within the created network by aiming for greedy routing paths with low stretch.

For our model, we analyze the existence of (approximate)-equilibria and the computational hardness in different underlying metric spaces. E.g., we characterize the set of equilibria in 1-2-metrics and tree metrics and show that Nash equilibria always exist. For Euclidean space, the setting which is most relevant in practice, we prove that equilibria are not guaranteed to exist but that the well-known $\Theta$-graph construction yields networks having a low stretch that are game-theoretically almost stable. For general metric spaces, we show that approximate equilibria exist where the approximation factor depends on the cost of maintaining any link.
\end{abstract}



\maketitle
\section{Introduction}
Many important real-world networks, like the Internet or peer-to-peer ad-hoc networks among smart devices, are created by the decentralized interaction of various agents that pursue their own goals, like maximizing their centrality, minimizing their latency, or maximizing the network throughput. These agents can thus be considered as acting selfishly and strategic, i.e., they try to optimize their costs for creating and using the network.
This observation sparked a whole research area devoted to game-theoretic network formation models~\cite{Papadimitriou01} and there is an abundance of interesting variants, typically called \emph{network creation games}, that capture how networks emerge in different domains, e.g., communication networks or social networks. 

Networks enable communication. To this end, it is commonly assumed that shortest paths are used for routing traffic within the created network. However, to make this work, all agents need to know the global structure of the network. Moreover, in case of dynamic changes to the network, this global view must be updated or otherwise shortest path routing would fail. In case of the Internet, shortest path routing is currently ensured by the usage of extensive routing tables that exactly specify which next-hop neighbor to use for which destination. All these routing tables must be maintained and updated accordingly, even for insignificant structural changes. 
This can be avoided by using  \emph{greedy routing}, where in each network node every incoming packet is simply routed to a neighbor that is closer to the packet's destination. For this, geographic information is needed, e.g., the nodes must be embedded in some underlying metric space and the positions in that space allow for deciding which next hop to use for routing. Exactly this has been proposed for the Internet~\cite{boguna2010sustaining}, i.e., to map the nodes of the Internet into a metric space such that greedy routing works. This is called a \emph{greedy embedding} and they can be computed efficiently~\cite{Blasius0KK20}.  

However, centrally computing a greedy embedding of a network does not account for the selfish strategic behavior of the network participants. The agents do not only want to ensure that greedy routing works, but at the same time they want to maximize their connection quality within the created network, e.g., they want to minimize their average stretch. 
Thus, strategic agents would "rewire" a given greedy embedding, if this reduces their stretch. However, this might endanger that greedy routing is guaranteed to succeed.

In this paper, we set out to investigate this tension between the stability of the network with respect to local structural changes by the agents and maintaining greedy routing.

\subsection{Model and Preliminaries}
We consider $n$ agents that correspond to points $\points = \{p_1,\dots,p_n\}$ in some metric space $\M = (\points,d_\M)$, where $d_\M(u,v)$ denotes the distance of $u\in \points$ and $v\in\points$ in the metric space. Besides arbitrary metric spaces where $d_\M$ has to satisfy the triangle inequality, we will also consider \emph{1-2-metrics}, where $d_\M(u,v) \in \{1,2\}$ for any $u,v\in\points$, \emph{tree-metrics}, where the distances are determined by a given weighted spanning tree $T$, such that $d_\M(u,v) = d_T(u,v)$, i.e., the distance between points $u,v\in\points$ is the length of the unique path between $u$ and $v$ in $T$, and the \emph{Euclidean-metric}, where the points are located in Euclidean space and the distance  is the Euclidean distance. We will omit the reference to the metric space $\M$ if it is clear from the context.

The goal of the agents is to create a directed network among themselves and for this each agent strategically decides over its set of outgoing directed edges. The \emph{strategy of agent~$u$} is $S_u \subseteq V\setminus \{u\}$, i.e., agent $u$ can decide to create edges to any subset of the other agents. Moreover, let $\s = (S_1,\dots,S_n)$ denote the \emph{strategy-profile}, which is the vector of strategies of all agents. As shorthand notation, for any agent $u\in V$ let $\s = (S_u,\s_{-u})$, where $\s_{-u}$ is the vector of strategies of all agents except agent $u$.
Any strategy-profile $\s$ uniquely defines a directed weighted network $G(\s) = \left(\points,E(\s),\ell\right)$,  where the edge-set $E(\s)$ is defined as $E(\s) = \bigcup_{u\in V}\{(u,v) \mid v\in S_u\}$ and the length of any edge $(u,v) \in E(\s)$ is equal to the distance of the positions of its endpoints in $\M$, i.e., it is defined as $\ell(u,v) = d_\M(u,v)$. 

Given a weighted directed network $G = (\points,E,\ell)$, where the nodes in $\points$ are points in metric space $\mathcal{M}$, a \emph{greedy path from \(u\) to \(v\) in $G$} is a path \(x_1, x_2, ... , x_j\), with $x_i\in \points$, for $1\leq i \leq j$, where \(x_1 = u\), \(x_j = v\), and $(x_i,x_{i+1}) \in E$, for $1\leq i \leq j-1$, such that \(d_\M(x_i, v) > d_\M(x_{i+1}, v)\) holds for all $1\leq i\leq j-1$. Thus, such a path is a directed path from $u$ to $v$ in $G$, where along the path the nodes get strictly closer to the endpoint of the path in terms of their distance in the ground space $\M$.
For two nodes \(u,v \in \points\), we define \(\gd{G}(u,v)\) as the length of the shortest greedy path between \(u\) and \(v\) in \(G\), where the length of a path \(x_1, x_2, ... , x_j\) is \(\sum_{i = 1}^{j-1} \ell(x_{i}, x_{i+1}) = \sum_{i = 1}^{j-1} d_\M(x_{i}, x_{i+1}) \).  If no greedy path exists between $u$ and $v$ in $G$, then $\gd{G}(u,v) = \infty$. We will call $\gd{G}(u,v)$ the \emph{greedy-routing-distance}\footnote{Note that the greedy-routing distance only depends on the network $G$ and the metric space $\M$ but not on any concrete greedy routing protocol. This more abstract definition ensures robust bounds, since our distance measure always gives a lower bound on the distance achieved by any particular greedy routing protocol.} between $u$ and $v$ in network~$G$. We say that \emph{greedy routing is enabled} in $G$, if any pair of nodes in $G$ has finite greedy-routing-distance. See \Cref{fig:greedy-stretch} for an example.
\begin{figure}[t]
\centering
\includegraphics[width=75mm,scale=0.5]{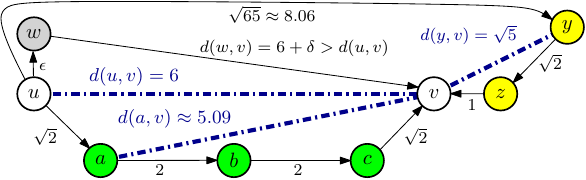}
\caption{Greedy paths in the Euclidean plane. Two such paths from $u$ to $v$ exist: $u,a,b,c,v$ and $u,y,z,v$. The latter shows, that even a single edge in a greedy path can be longer than $d(u,v)$.  
Path $u,w,v$ is not a greedy path, since $d(w, v) > d(u,v)$. The shortest greedy path is $u,a,b,c,v$, so we have $\gd{G}(u,v) =2\sqrt{2}+4 \approx 6.83$. This yields $\text{stretch}_{G}(u, v) = \gd{G}(u,v)/ d(u,v) =\frac{2\sqrt{2}+4}{6} \approx 1.14$. Note that for small enough $\varepsilon$ and $\delta$ the path $u,w,v$ has length less than $\gd{G}(u,v)$. Thus, the shortest $u$-$v$-path in the network may not be a greedy path.} 
\label{fig:greedy-stretch}
\end{figure}

For any two nodes $u$ and $v$ in network $G$, we will compare their greedy-routing-distance with their distance in the ground space $\M$. The ratio of these values is called the \emph{stretch}, i.e., we have 
\[\text{stretch}_{G}(u, v) =
\begin{cases}
    \frac{\gd{G}(u, v)}{d_\M(u, v)} &\text{, if a greedy path from $u$ to $v$ exists in $G$ },\\
    Z &\text{, otherwise,}
\end{cases}\]
where \(Z\) is some arbitrarily large number that serves as a penalty for not having a greedy path. Intuitively, the stretch measures the detour that the best greedy path has to take, compared to the shortest possible path, i.e., to having a direct edge to the target node.

Agents choose their strategy to minimize their \emph{cost} within the formed network. The cost of agent~\(u\) in network $G(\s)$ is defined as
\[c_u(\s) = \text{stretchcost}_u(\s) + \text{edgecost}_u(\s),\]
where \(\text{stretchcost}_u(\s) = \sum_{v \in \points\setminus \{u\}}  \text{stretch}_{G(\s)}(u, v)\) and \(\text{edgecost}_u(\s) = \alpha |S_u|\), for a given $\alpha > 0$. The latter is a global parameter and allows to adjust the comparison of edge costs versus stretch costs.
The \emph{social cost} of a network $G(\s)$ is defined as $c(\s) = \sum_{u\in \points}c_u(\s)$. For any set of points $\points$, the network $G(\s^*) = (\points,E(\s^*), \ell)$ minimizing the social cost is called the \emph{social optimum network} for~$\points$.

An \emph{improving response} for an agent \(u\) for a strategy-profile \(\s = (S_u,\s_{-u})\) is a strategy \(S'_u\) such that \(c_u((S'_u,\s_{-u})) < c_u((S_u,\s_{-u}))\), i.e., by employing strategy $S'_u$, agent $u$ has strictly lower cost compared to using strategy $S_u$.
A strategy $S^*_u$ is called a \emph{best response} for agent~\(u\) for strategy-profile $\s = (S_u,\s_{-u})$, if $c_u(S^*_u,\s_{-u}) \leq c_u(S'_u,\s_{-u})$ for any other strategy $S'_u \subseteq V\setminus\{u\}$, i.e., a best response strategy minimizes agent $u$'s cost, given that the strategies of the other agents are fixed. 

A strategy-profile \(\s\) is in \emph{pure Nash Equilibrium} (NE) if no agent has an improving move for $\s$, i.e., in $\s$ every agent already employs a best response. Since we have bijection between strategy-profiles $\s$ and the corresponding networks $G(\s)$, we will say that network $G(\s)$ is in NE, if $\s$ is in NE. A network~\(G(\s)\) is in \emph{Greedy Equilibrium} (GE) \cite{lenznerGreedy2012} if no agent has an improving response that differs from its current strategy in adding, swapping or deleting a single incident outgoing edge, where a swap is a combination of deleting an incident outgoing edge and adding another one.  
Notice that any network in NE is also in GE. 
A network \(G(\s)\) is in \(\beta\text{-approximate}\) NE ($\beta$-NE) if no agent \(u\) can change its strategy such that its cost decreases below \(\frac{1}{\beta} c_u(\s)\), i.e., no agent has an improving response that reduces its cost by at least a factor of $\beta$. 

An \emph{improving (best) response path} is a sequence of strategy-profiles $\s_0,\s_1,\s_2,\dots,\s_k$, such that $\s_i$ results from some agent changing to an improving (best) response in $\s_{i-1}$, for $1\leq i\leq k$. An \emph{improving (best) response cycle} (IRC or BRC, respectively) is a cyclic improving response path, i.e., where $\s_0 = \s_k$ holds.
By definition, every best response path (cycle) is also an improving response path (cycle). The non-existence of improving response cycles, i.e., if every improvement path has finite length, is equivalent to the existence of an ordinal potential function~\cite{monderer1996potential}. The latter implies that NE can be found via natural sequential improvement dynamics. A strategic game is called \emph{weakly acyclic (weakly acyclic under best response)} if from every strategy vector \(\s\) there exists a finite improving (best) response path that starts in $\s$ and ends in a NE. 

\subsection{Discussion of the Model Assumptions} 
Our main goal is to consider a simple but general network creation model that features agents that use geographic information of their neighbors for routing.

\textbf{Considering Different Metrics:} For capturing as many realistic settings as possible, we consider different underlying metrics, starting from the simplest non-trivial metric, 1-2-metrics, that capture settings where pairwise communication can only have low quality (length 2) or high quality (length 1), e.g., due to hardware constraints. Such metrics have often been studied, e.g., for the TSP \cite{karpReducibility1972, adamaszekNew2018}.  
We also consider tree-metrics that are commonly used for approximating other metrics~\cite{FakcharoenpholRT04} and that naturally emerge if overlay networks with an underlying tree topology are considered. The setting of Euclidean metrics is natural for ad-hoc networks~\cite{RamanathanR02} but also for classical wireline networks.

\textbf{Using Shortest Greedy Routing Paths:} Given that an abundance of different greedy routing protocols exist~\cite{MauveWH01}, we present a more general analysis by abstracting away from any concrete greedy routing protocol. To achieve this, our definition of the greedy-routing distance only depends on the network $G$ and the metric space $\M$ but not on any concrete assumption of which neighbor that is closer to the target is used as next hop in routing. This more abstract definition ensures robust bounds, since our distance measure always gives a lower bound on the distance achieved by any particular greedy routing protocol. E.g., always using the neighbor that is closest to the target can yield greedy routing paths that are arbitrarily longer, compared to the shortest greedy routing path.
Thus, our greedy routing can be understood as the result of iteratively improving the employed greedy routing paths over time\footnote{Such a process cannot get stuck in local non-greedy routing since we assume all-pairs communication, i.e., at any time every network node ensures to have a greedy routing path to all other nodes. Thus, a local non-greedy routing would directly urge the corresponding node of a routing dead-end to change its edge-set.}.
Note that the shortest greedy routing path can be found efficiently via Dijkstra's algorithm.

\textbf{Considering the Stretch:} From Moscibroda, Schmid, and Wattenhofer~\cite{moscibrodaTopologies2006} we adopt measuring the connection quality between agents via their stretch. This assumes that by using the network, agents can infer their network distance to any other node, e.g., by measuring their latency. We note that with slight modifications of the respective proofs all our results carry over if we would directly use the distance in the network, as in~\cite{biloGeometric2019}, without dividing by the respective distance in the underlying metric space.

\textbf{The Role of the Parameter $\alpha$:} As is common also in other variants of network creation games~\cite{fabrikantNetwork2003,moscibrodaTopologies2006,biloGeometric2019}, the parameter $\alpha$ captures the trade-off between the costs for establishing a link in the network and the costs for communication between two nodes, e.g., the observed latency. We can assume a high alpha for settings where the edge costs dominate the stretch costs and a low alpha for the inverse setting. This allows to model a wide range of applications, e.g., the parameter $\alpha$ could encode the basic set-up costs for establishing a cryptographically secure connection, or for purchasing hardware, e.g., directed antennae or cables. Given these examples and for the sake of simplicity, we assume that the same parameter $\alpha$ applies to all agents.

\subsection{Related Work}
Network Creation Games (NCGs) were first introduced by~\citet{fabrikantNetwork2003}. In their model, agents that correspond to nodes of a network strategically create incident undirected edges with the goal of minimizing their sum of hop-distances to all other agents. 
They show that NE always exist, that the social optimum network is either a clique or a star,
and that computing a best response is NP-hard.  
NCGs admit IRCs~\cite{KL13} and any network in GE is in $3$-NE~\cite{lenznerGreedy2012}.   
Moreover, many variations of NCGs have been investigated, e.g., a bilateral version

~\cite{CP05,FGLZ23}, variants with locality~\cite{BiloGLP16,CL15}, with robustness~\cite{MeiromMO15,CLMM16,EFLM20}, with budget constraints~\cite{LaoutarisPRST08,EhsaniFMSSSS11}, with non-uniform edge prices~\cite{MeiromMO14,CLMM17,BiloFLLM21}, with heterogeneous agents~\cite{BullingerLM22}, or with temporal edges~\cite{tempNCG}.

Closer to our model are NCGs that involve geometry. In the wireline strong connectivity game~\cite{eidenbenzEquilibria2006} agents that correspond to points in the plane create undirected edges to ensure that they can reach all other agents. 
There, NE exist and can be found efficiently. Geometric spanner games~\cite{abamGeometric2019} are similar, but the agents want to ensure a given maximum stretch.   
\citet{biloGeometric2019} defined a similar model with undirected edges, where the agents want to minimize the sum of their shortest path distances to all other agents. 
 
They show that NE exist in 1-2-metrics and tree-metrics and that every GE is in 3-NE for any metric. Also, even on 1-2-metrics computing a best response is NP-hard. 
Later, it was shown that $(\alpha+1)$-NE always exist~\cite{friedemannEfficiency2021}. 

Closest to our model is the work by \citet{moscibrodaTopologies2006}. Their model involves agents that form directed edges and that have the same cost function as in our model. However, classical shortest paths are used for defining the stretch. For Euclidean metrics they provide an instance that does not admit a NE and they show that deciding NE existence is NP-hard. Note that networks in NE might not be in NE in our model an vice versa. Thus, none of their results can be directly transferred to our model.

For background on (Algorithmic) Game Theory and Multiagent Systems, we refer to the textbooks by \citet{AGT-book} and \citet{MAS_book}.

{
\begin{table*}[t]
    \centering
    \caption{Result overview. Abbreviations used: IRC, BR(C) (Improving/Best Response (Cycle)), BR-WA (Weakly Acyclic under Best Response), NE-ver (NE-verification), \{NE,GE\}-dec (\{NE,GE\}-decision).} \resizebox{\textwidth}{!}{
    \begin{tabular}{lrccc}\toprule
        \textbf{} & \textbf{1-2 Metrics} & \textbf{Tree Metrics} & \textbf{Euclidean Metrics} & \textbf{General Metrics} \\ \midrule
        \textbf{NE Existence} & NE always exist & NE always exist & no existence & no existence \\
        & NE characterization & {GE \& NE unique} \\
        \midrule
        \textbf{Dynamics} & IRCs exist, no BRCs & BR-WA & IRCs+BRCs & IRCs+BRCs
                \\
        \midrule
        \textbf{Complexity} & $\alpha > 1/2$: BR NP-hard   & BR NP-hard & BR NP-hard & BR NP-hard \\
          & NE-ver NP-hard &&&\\
        &$\alpha \leq 1/2$: BR in P  & NE-dec in P &  & NE-dec NP-hard\\
        &  NE-ver in P && & GE-dec NP-hard\\
        \midrule
       \textbf{Approx-NE} & $\mathcal{O}(\log n)$& $1$ & $5$ & $\alpha +1$ \\ \bottomrule
    \end{tabular}}
    \label{tab:overview}
\end{table*}
} 
\subsection{Our Contribution}
We propose and analyze the first network creation game where intelligent selfish agents aim for optimizing their communication quality while at the same time maintaining that all-pairs greedy routing is enabled. Given the favorable properties of greedy routing for dynamically changing networks among smart devices, this is highly relevant for current and future multi-agent network formation scenarios. 

As our main technical contribution, we uncover the influence on the underlying metric space on the existence of equilibria, the game dynamics, and the computational complexity of computing best response strategies and on deciding equilibrium existence. See \Cref{tab:overview} for an overview. 

We show that NE always exist on 1-2 metrics and tree metrics but on Euclidean metrics there are instances that do not admit NE. However, we show that in all considered metrics approximate-NE with low approximation factor do exist and can be constructed efficiently. 
Regarding the dynamics, we show that improving response cycles exist for almost all variants but that on 1-2 metrics and tree metrics convergence to an NE is guaranteed if only best responses are played. On the complexity side, we show that computing such best responses is NP-hard in most cases. Also deciding NE existence is NP-hard for some cases.   

In comparison to \cite{moscibrodaTopologies2006}, where shortest path routing instead of greedy routing is used, our results on Euclidean metrics are similar but this is not obvious at all. Adding the constraint of greedy routing forces the agents to create certain edges which yields entirely different behavior\footnote{For example, in a metric space where all the nodes have the same distance to each other and $\alpha=1$, in the model of \cite{moscibrodaTopologies2006} a star is a NE while in our model the only NE is the complete graph.}. Thus, their results do not carry over to our model. Moreover, we show the even stronger result that not even GE are guaranteed to exist. 

We emphasize, that the efficient construction of approximate equilibria for Euclidean and general metrics directly carries over to the model from~\citet{moscibrodaTopologies2006}. This is true, since our agents use greedy routing paths instead of shortest paths, so the path choice is more restricted in our model. Thus, the stretch can only decrease when considering shortest paths. Hence, we present the first positive results on equilibrium existence for the model in~\cite{moscibrodaTopologies2006}. Also, our constant approximate NE for Euclidean metrics via $\Theta$-graphs are highly relevant in practice. $\Theta$-graphs are well-known geometric spanners with many applications.

\section{1-2 Metrics}\label{sec:12metrics}
We consider metrics $\M$ where for every pair of nodes \(u, v\) we have \(d_\M(u,v) \in \{1,2\}\). We denote edges of weight 1 and 2 as \emph{1-edges} and \emph{2-edges}, respectively. 1-2-metrics is the simplest class of non-trivial metric spaces and they have often been studied, e.g., for the TSP \cite{karpReducibility1972, adamaszekNew2018}.

In spaces where all edges have the same length, direct edges would be the only possible greedy paths and as such the agents' strategies would be independent of each other. We start by giving some general statements about greedy paths and agent strategies in 1-2-metrics.

\begin{restatable}{lemma}{lemmastretches}
    \label{lem:12-stretches}
    For any network in a 1-2-metric, all greedy paths have stretch  \(1\) or \(\frac{3}{2}\) and consist of at most two edges.
\end{restatable}
\begin{proof}
    Let \(u, v \in \mathcal{P}\) and \(u \neq v\). If \(d(u,v) = 1\), there can be no other node closer to \(v\) than \(u\). Thus, the only possible greedy path is \(u,v\) with a stretch of \(1\) that uses one edge. Otherwise, \(d(u,v) = 2\), and only nodes \(x\) with \(d(x,v) = 1\) are closer to \(v\) than \(u\) and no other node is closer to \(v\) than any such \(x\). Thus, the only possible greedy paths are \(u,v\) and \(u,x,v\) for any such \(x\). These paths use one or two edges, respectively, and have stretches of \(1\) and either \(1\) or \(\frac{3}{2}\) depending on \(d(u,x)\).
\end{proof}

When examining these possible greedy paths, we also directly get the following statement.

\begin{remark}
    \label{lem:12-2edges-first}
    For any network in a 1-2-metric, \(2\)-edges can only be the first edge of any greedy path.
\end{remark}
Also, we find that 1-edges are crucial for enabling greedy routing.
\begin{restatable}{lemma}{lemmaedges}
    \label{lem:12-1edges}
    For any network in a 1-2-metric, if greedy routing is enabled then all \(1\)-edges are built in both directions.
\end{restatable}
\begin{proof}
	For the sake of contradiction, let \(u, v \in \mathcal{P}\) be such that \(d(u,v) = 1\) but there is no edge between \(u\) and \(v\) that \(u\) can use as greedy path, i.e., agent~\(u\) does not build an edge to \(v\). However, without such an edge no greedy path from \(u\) to \(v\) can exist because no other node can be closer to \(v\) than \(u\) itself, a contradiction.
\end{proof}

\subsubsection*{\textbf{Equilibrium Existence}} We start by showing the equivalence of NE and social optima.
\begin{restatable}{theorem}{thmtwoone}
    \label{the:12-directed-pne-opt}
    In a 1-2-metric, every NE is a social optimum and every social optimum is a NE.
\end{restatable}
\begin{proof}
    First, by \Cref{lem:12-1edges}, every outgoing 1-edge incident to an agent \(u \in \mathcal{P}\) needs to be part of all of its best responses. Second, by \Cref{lem:12-2edges-first}, no 2-edge that any agent builds can be part of a greedy path that is used by any other agent. Thus, the strategy of an agent $u \in \mathcal{P}$ does not influence the cost function of any other agent $v \in \mathcal{P}\setminus\{u\}$ beyond the 1-edges that need to be part of every strategy with costs less than \(Z\). Using this, we show by contradiction that every NE is a social optimum and vice versa.
    
    If there was a NE that is not a social optimum, the latter would have lower social cost than the NE and thus, some agent~\(u\) must have lower cost in the social optimum compared to its cost in the NE. Hence, it could reduce its cost by deviating to its strategy in the social optimum network.
    
    In the same vein, if there was a social optimum network that is not in NE, an agent \(u\) would have an improving move that would also improve the social cost.
\end{proof}

This implies that NE always exist. In the following we characterize all NEs. For this, we introduce \emph{Domination Set Graphs}, which are based on the notion of a \emph{directed dominating set}~\cite{fu1968dominating} of a directed graph $G = (V,E)$. For any node $u\in V$, let  $W_1(u), W_{1\to1}(u), W_2(u)$ denote the sets of nodes that node~$u$ can reach by a 1-edge, by a path of two 1-edges, and by a 2-edge, respectively.  We denote by $W_2^-(u)$ the set of nodes that $u$ builds a 2-edge to and that $u$ cannot reach via a path of two 1-edges, i.e. $W_2^-(u)=W_2(u)\setminus W_{1\to1}(u)$, and let $W_{2}^+(u)$ be the set of nodes from $W_2(u)$ to which $u$ also has a path of two 1-edges, i.e., $W_2^+(u) = W_2(u) \cap W_{1\to1}$. Note that $W_2(u) = W_2^-(u) \cup W_2^+(u)$ and $W_2^-(u) \cap W_2^+(u) = \emptyset$.
The \emph{out-neighborhood} of $u$ is the set $N(u)=W_1(u)\cup W_2(u)$.

Subset $D \subseteq V$ is \emph{dominating}, if for every node $v \in V\setminus D$ there exists a node $u \in D$ such that $(u,v) \in E$. For this we say that \emph{$u$ dominates $v$} and that $u$ dominates its out-neighborhood.
\begin{definition} Let $G_{-u}^1 = \left(\points\setminus\{u\},E^1_{-u}\right)$ be the network without node~$u$ that only contains all the 1-edges, i.e., we have that $E^1_{-u} = \left\{(v,w) \in\left(\points \setminus \{u\}\right)^2\ \mid d(v,w) = 1\right\}$.

	A \emph{Domination Set Graph} (DSG) on a 1-2-metric \((\mathcal{P}, d)\) is a network \(G = (\mathcal{P}, E)\)  with: 
	\begin{enumerate}	
		\item[(i)] \(E \supseteq \{(v,w) \in \mathcal{P}^2\ |\ d(v,w) = 1\}\),
		\item[(ii)] for every node \(u\), its out-neighborhood must be dominating in the network $G_{-u}^1$, and
		\item[(iii)] \( |W_{2}^+(u)| \) cannot be decreased by a single edge deletion or swap without increasing $|N(u)|$.
	\end{enumerate}
\end{definition}

By (i), a DSG includes all 1-edges, by (ii), for every node \(u\), its out-neighborhood is dominating in the network without $u$ consisting of only the 1-edges. Condition (iii)  states that in a DSG an edge to a node $v\in W_2^+(u)$ cannot be removed or swapped with another one to a node $w\notin W_{1\to 1}(u)$ without resulting in the out-neighborhood of $u$ no longer being dominating in the network $G_{-u}^1$.

We now define more specific DSGs:
\begin{definition}
	A \emph{Minimum Domination Set Graph} (MinDSG) is a DSG in which for every node $u$ its out-neighborhood contains all nodes reachable via 1-edges and the minimum number of nodes reachable via 2-edges so that it is dominating in $G_{-u}^1$. Thus, all out-neighbor\-hoods are minimum size dominating sets given that they need to include all nodes reachable via 1-edges.
\end{definition}
\begin{definition}
	A \emph{Maximum Domination Set Graph} (MaxDSG) is a DSG that contains all possible edges, except 2-edges $(u,w)$, where $(u,v)$ and $(v,w)$ both are 1-edges.
\end{definition}
Note that MaxDSGs are unique and the stretch between any pair of nodes in a MaxDSG is \(1\). 
\begin{definition}
\emph{A Balanced Domination Set Graph for $\alpha > 0$ (BDSG$(\alpha)$}), is a subclass of Domination Set Graphs, where for any node $u$, the quantity \( \left( \alpha-\frac{1}{2}\right)|W_{2}(u)|+\frac{1}{2}|W_{2}^+(u)| \) is minimized, achieving a balance between the number of 2-edges that $u$ builds in total, and the number of 2-edges to nodes in $W_{1\to1}(u)$. Note that for $\alpha<\frac{1}{2}$, the BDSG corresponds to the MaxDSG.
\end{definition}
See \Cref{fig:DSG} for examples. Next, we show that DSGs are useful for our analysis.
\begin{figure}[h]
\includegraphics[width=\textwidth]{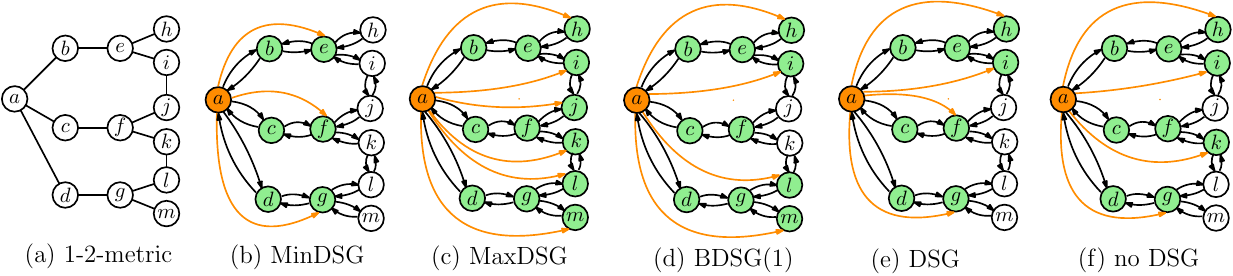}
\caption{Examples of DSGs ((b)-(e)) where only 1-edges and outgoing edges of node $a$ (colored orange) are shown, agent $a$ has stretch 1 to green colored nodes. (a) Sample 1-2-metric where only 1-edges are shown; (b) MinDSG for~(a), 1-edges are colored black and 2-edges are colored orange; (c) the unique MaxDSG for~(a); (d) BDSG(1) for~(a); (e) some other DSG for~(a); (f) a network that is not a DSG, since condition (iii) is violated: $|W_2^+(d)|$ can be decreased without increasing $|N(d)|$, by swapping the edge $(a,g)$ to $(a,m)$. Note that for node $a$ the set $\{b,c,d,e,f ,g\}$ is a minimum size dominating set for $G_{-a}^1$, but this set cannot be agent~$a$'s strategy in a BDSG since $0.5|W_2(a)|+ 0.5|W_2^+(a)|$ is not minimal.}
\label{fig:DSG}
\end{figure}
\begin{lemma}
    \label{lem:12-DSG}
    In a 1-2-metric any network in GE is a DSG.
\end{lemma}

\begin{proof}
    By \Cref{lem:12-1edges}, all 1-edges need to be built in a GE network and thus condition (i) of a DSG is met.
    Next, assume for the sake of contradiction, that condition (ii) of a DSG is violated and that the set \(N(u)\) of nodes that some agent \(u\) builds edges to is not dominating in $G_{-u}^1$. In other words, there is some node~\(v\) such that neither it nor any of its in-neighbors are in \(N(u)\). Thus, the shortest path from \(u\) to \(v\) has at least \(3\) edges. By \Cref{lem:12-stretches}, no such path can be a greedy path and thus agent~\(u\) could improve its cost by building an edge to \(v\).
    Finally, we consider condition (iii) that states that for any agent $u$, the size of \(|W_{2}^+(u)| \) cannot be decreased by a single operation without increasing $|N(u)|$, provided that $N(u)$ is dominating in $G_{-u}^1$. For the sake of contradiction, let us assume that this condition is violated in a GE network. There are two cases that this can happen. In the first case, 
    some agent~$u$ removes an edge to a node $v\in W_2^+(u)$, while agent~$u$'s out-neighborhood is still dominating in $G_{-u}^1$. Thus, agent~$u$ could reduce its edgecosts by removing this edge, while its stretchcosts are not increased since the stretch to $v$ is already 1 (there is a path of two 1-edges from $u$ to $v$) and the stretches to every other node remain the same. In the second case, some agent $u$ swaps an edge to a node $v\in W_2^+(u)$ with one to a node $w\notin (W_2(u)\cup W_{1\to1}(u))$, while agent~$u$'s out-neighborhood is still dominating in $G_{-u}^1$. Thus, agent $u$ could reduce its stretchcost, since the stretch to $w$ would become 1 instead of $\frac{3}{2}$, while its edgecosts and the stretches to every other node remain the same. Thus, we have that every GE is a DSG.
\end{proof}

\begin{lemma}
    \label{lem:12-DSG-greedy}
    In a 1-2-metric, in any DSG greedy routing is enabled.
\end{lemma}

\begin{proof}
    Every agent \(u\) trivially has a greedy path to any node that it builds an edge to. For any node \(v\) that \(u\) does not build an edge to, there is still a node \(x\) in distance \(1\) from \(v\) that \(u\) builds an edge to because the out-neighborhood of $u$ is by definition a dominating set in $G_{-u}^1$. Thus, the path \(u,x,v\) exists in the DSG and it is a greedy path because \(d(u,v) = 2 > 1 = d(x,v)\).
\end{proof}

We now show that DSGs can be used to characterize NE.
\begin{restatable}{theorem}{thmtwofive2}
    \label{lem:12-dir-NE-BDSG}
    In a 1-2-metric, for any $\alpha >0$, all BDSG$(\alpha)$ are the only NE.
\end{restatable}
\begin{proof}
Let $V_1(u) = \{v \in \points\mid d_\M(u,v) = 1\}$ and $V_2(u) = \{v\in \points\mid d_\M(u,v) = 2\}$ be the sets of nodes in distance~1 and in distance~2 from agent~$u$ in the ground space. Also, let $V_2^+(u) = V_2(u) \cap W_{1\to1}(u)$ denote the set of nodes in distance~2 from $u$ that can be reached by a path of two 1-edges, and $V_2^-(u) = V_2(u) \setminus V_2^+(u)$ the set of nodes in distance~2 from agent $u$ that $u$ cannot reach via a path of two 1-edges.
Thus, we have $V_2(u)=V_2^+(u)\cup V_2^-(u)$. The cost of any strategy-profile $\s$ for agent~$u$ is equal to
$$c_u(\s)=|V_1(u)|+|V_{2}^+(u)|+|W_2^-(u)|+\frac{3}{2}\left[ |V_2^-(u)|-|W_2^-(u)|\right] +\alpha \left(|W_2^+(u)|+|W_2^-(u)|+|V_1(u)|\right).$$
By rearranging terms, we get:
\begin{equation}
c_u(\s) = (\alpha+1)|V_1(u)|+|V_{2}^+(u)|+\left(\alpha-\frac{1}{2}\right)|W_2^-(u)|+\frac{3}{2}|V_2^-(u)|+\alpha |W_2^+(u)|. \label{eq:cost_u}
\end{equation}

Now, let $\overline{\s}$ denote the strategy-profile of a BDSG$(\alpha)$ network $G(\overline{\s})$. Assume towards a contradiction that there is an agent $u\in \points$ that could decrease its cost in $G(\overline{\s}) = G((\overline{S_u},\overline{\s}_{-u}))$ by deviating to some other strategy $S_u \neq \overline{S_u}$. Let $\s = (S_u,\overline{\s}_{-u})$ denote the strategy-profile after the deviation of agent~$u$. Note that, by \Cref{lem:12-1edges}, also in $\s$ all 1-edges have to be built, in particular, all outgoing 1-edges of agent~$u$, since otherwise greedy routing would not be enabled and thus agent $u$ could not decrease its cost.
Let $V_1(u), V_2^+(u), V_2^-(u), W_2^+(u)$ and $W_2^-(u)$ denote the sets from equation~(\ref{eq:cost_u}) with respect to strategy profile $\s$ and let $\overline{V_1(u)}, \overline{V_2^+(u)}, \overline{V_2^-(u)}, \overline{W_2^+(u)}$ and $\overline{W_2^-(u)}$ be the sets from equation~(\ref{eq:cost_u}) for strategy-profile $\overline{\s}$. Also, remember that $W_2(u) = W_2^+(u) \cup W_2^-(u)$ and let $\overline{W_2(u)} = \overline{W_2^+(u)} \cup \overline{W_2^-(u)}$. Now,
observe that $V_1(u) = \overline{V_1(u)}$, $V_2^+(u) = \overline{V_2^+(u)}$, and $V_2^-(u) = \overline{V_2^-(u)}$, since these sets only depend on the 1-edges built by agent $u$, which are the same in $\s$ and $\overline{\s}$.
Thus, equation~(\ref{eq:cost_u}) gives agent $u$'s cost in strategy-profile $\s$. For strategy-profile $\overline{\s}$, equation~(\ref{eq:cost_u}) yields
\begin{equation}
c_u(\overline{\s}) = (\alpha+1)|V_1(u)|+|V_{2}^+(u)|+\left(\alpha-\frac{1}{2}\right)|\overline{W_2^-(u)}|+\frac{3}{2}|V_2^-(u)|+\alpha |\overline{W_2^+(u)}|, \label{eq:cost_u_minDSG}
\end{equation}
Thus, using equation~(\ref{eq:cost_u}) and equation~(\ref{eq:cost_u_minDSG}), that $|W_2(u)| = |W_2^+(u)| + |W_2^-(u)|$, that $|\overline{W_2(u)}| = |\overline{W_2^+(u)}| + |\overline{W_2^-(u)}|$, and that agent~$u$ has strictly lower cost in $\s$ compared to $\overline{\s}$, we have
\begin{align*}
 c_u(\s)&< c_u(\overline{\s})\\  
 \left(\alpha-\frac{1}{2}\right)|W_2^-(u)|+\alpha |W_2^+(u)|&< \left(\alpha-\frac{1}{2}\right)|\overline{W_2^-(u)}|+\alpha |\overline{W_2^+(u)}|\\
 \left(\alpha-\frac{1}{2}\right) \left(|W_2(u)|-|W_2^+(u)|\right)+\alpha |W_2^+(u)|&<  \left(\alpha-\frac{1}{2}\right) \left(|\overline{W_2(u)}|-|\overline{W_2^+(u)}|\right)+\alpha |\overline{W_2^+(u)}|\\
  \left(\alpha-\frac{1}{2}\right)|W_2(u)|-\left(\alpha-\frac{1}{2}\right)|W_2^+(u)|+\alpha |W_2^+(u)|&<  \left(\alpha-\frac{1}{2}\right)|\overline{W_2(u)}|-\left(\alpha-\frac{1}{2}\right)|\overline{W_2^+(u)}|+\alpha |\overline{W_2^+(u)}|\\
  \left(\alpha-\frac{1}{2}\right)|W_2(u)|+\frac{1}{2}|W_2^+(u)|&< \left(\alpha-\frac{1}{2}\right)|\overline{W_2(u)}|+\frac{1}{2}|\overline{W_2^+(u)}|.\\
\end{align*}
This is a contradiction, because according to the definition of a BDSG$(\alpha)$, in $G(\overline{\s})$ no agent $u$  can decrease the quantity $b = (\alpha-0.5)|\overline{W_2(u)}|+0.5|\overline{W_2^+(u)}|$. However, strategy $S_u$ of agent $u$ strictly decreases $b$.
Therefore, a BDSG$(\alpha)$ is a NE for any $a>0$, and it is the only one.
\end{proof}
\begin{restatable}{theorem}{thmtwofive}
    \label{lem:12-dir-sum-const}
    In a 1-2-metric we have that (i) for \(\alpha < \frac{1}{2}\), the MaxDSG is the only NE, (ii) for \(\alpha = \frac{1}{2}\), all DSGs are the only GE, but only the DSGs where no node builds edges to nodes reachable by two 1-edges are the only NE, (iii) for \(\alpha > \frac{1}{2}\), all MinDSG are GE. 
\end{restatable}
\begin{proof}
First, let \(\alpha < \frac{1}{2}\). From \Cref{lem:12-dir-NE-BDSG} we have that $BDSG(\alpha)$ are the only NE. When $\alpha < \frac{1}{2}$, the network $BDSG(\alpha)$ corresponds to the MaxDSG. Thus,
the MaxDSG is the only NE.

Next, let \(\alpha = \frac{1}{2}\). Then, the network $BDSG(\alpha)$ corresponds to DSGs where for every node $u$ the quantity $|W_2^+(u)|$ is minimized, i.e. no node builds edges to nodes that are reachable by two 1-edges. Thus, this category of DSGs are the only NE.

Regarding GE, every  DSG is a GE because removing any edge would either remove all greedy paths to some node or increase the stretch to the endpoint of the edge by at least \(\frac{1}{2} = \alpha\). Also adding any edge would decrease stretches by at most \(\frac{1}{2} = \alpha\). Swapping an edge also cannot decrease stretches. The number of edges an agent $u$ builds remains the same after a swap. Consequently, agent $u$ can reduce its cost only by reducing its stretchcosts. The stretchcost of agent~$u$ is equal to
\begin{align}
    \frac{3}{2}\left(|V|-1\right)-\frac{1}{2}\left(|W_{1}(u)|+|W_{1\to1}(u)|+|W_{2}^-(u)|\right)
    \label{1-2swap}
\end{align}
since, by Lemma \ref{lem:12-stretches}, all greedy paths have stretches $1$ or $\frac{3}{2}$.
By Lemma \ref{lem:12-1edges}, all 1-edges have to be built to enable greedy routing, thus agent $u$ can only swap 2-edges.  

If after one swap the size of $|W_{2}^-(u)|$:
\begin{itemize}
    \item decreases, then the stretchcost of agent~$u$ increases because of (\ref{1-2swap});
    \item remains the same, then the stretchcost of agent~$u$ will remain the same; 
    \item increases, that means that $|W_{2}^+(u)|$ will be decreased. By the definition of a DSG, the value $|W_{2}^+(u)|$ cannot be decreased without increasing the size of the out-neighborhood of $u$. Since we are considering swaps, the number of edges that $u$ builds cannot change, a contradiction.
\end{itemize}
Therefore, swapping an edge is not an improving move for agent $u$.
Since, by \Cref{lem:12-DSG}, all GE are DSG, there can be no other GE.

    Next, let \(\alpha > \frac{1}{2}\). Every MinDSG is a GE because removing an edge would also remove all greedy paths to some node and adding an edge would decrease stretches by at most \(\frac{1}{2} < \alpha\). As we showed before, swapping an edge is not an improving move either.
\end{proof}

\subsubsection*{\textbf{Dynamic Properties}}
We explore if NE networks can be found by allowing agents iteratively to select improving  or best responses. 
First, we show that improving response cycles (IRCs) exist. Then we contrast this with the positive result that best response cycles (BRCs) cannot exist. The proof of the former relies on the absence of greedy paths between certain agents, resulting in costs higher than \(Z\). This is necessary, since no IRCs can exist if the costs of all agents are less than \(Z\), i.e., if the network already enables greedy routing. This can be proven with the same arguments used below to show that best response cycles cannot exist. 

\begin{restatable}{theorem}{thmtwosix}
\label{the:12-dir-irc}
    In a 1-2-metric, IRCs can exist for any $\alpha >0$.
\end{restatable}

\begin{proof}
 We show that there is an IRC in the metric shown in \Cref{fig:12-irc}a regardless the value of \(\alpha\). In particular, in  each step of the IRC, the  stretchcosts of the agents that changes its strategy improves by exactly \(Z - 1\). By definition of \(Z\), this will always outweigh any changes in the edgecosts, since agents always build the same number of edges of the same length.

Consider the network shown in \Cref{fig:12-irc}b. If agent \(u\) changes its strategy from $\{c,d,i\}$ to $\{a,b,v\}$ we get the network shown in \Cref{fig:12-irc}c. With this agent $u$ only loses greedy paths of length 1 to the three nodes \(c,d\) and \(i\) and of length 2 (and stretch 1) to node~\(j\) but it gains greedy paths of length 1 to the three nodes \(a,b\) and \(v\) and of length 2 (and stretch 1) to the two nodes \(e\) and \(f\). Thus, agent \(u\)'s stretchcost improves by \(Z - 1\).
\begin{figure}[h]
\includegraphics[width=\textwidth]{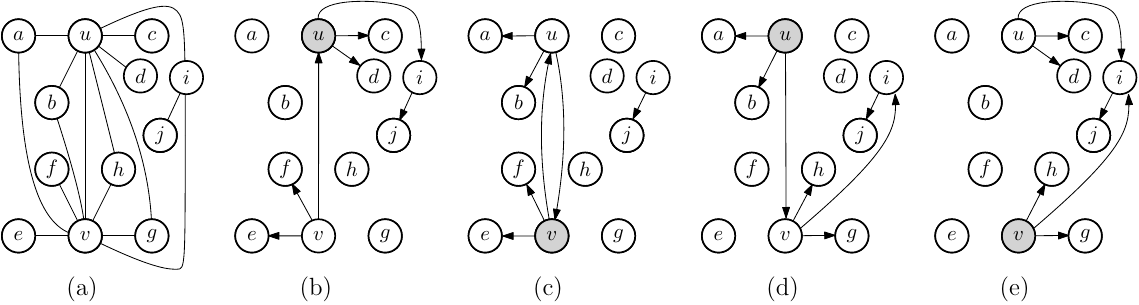}
\caption{1-2 metric with improving response cycle. (a) sample 1-2-metric, only 1-edges are shown; (b) - (e): steps of the improving response cycle, shaded agents change their strategy in the respective step.}
\label{fig:12-irc}
\end{figure}

In the network depicted in \Cref{fig:12-irc}c, if agent~\(v\) changes its strategy from $\{e,f,u\}$ to $\{h,g,i\}$ we get the network shown in~\Cref{fig:12-irc}d. With this agent $v$ only loses greedy paths of length 1 to the three nodes \(e,f\) and \(u\) (note that the paths $v,u,a$ and $v,u,b$ are not greedy paths) but at the same time it establishes greedy paths of length 1 to the three nodes \(g, h\) and \(i\) and of length 2 (and stretch~1) to \(j\). Thus, agent~\(v\)'s stretchcosts improve by \(Z - 1\).

In the network shown in \Cref{fig:12-irc}d, if agent~\(u\) changes its strategy $\{a,b,v\}$ to $\{c,d,i\}$ we get the network shown in~\Cref{fig:12-irc}e. With this, agent $u$ loses greedy paths of length 1 to the three nodes \(a,b\) and \(v\) but it gains greedy paths of length 1 to the three nodes \(c,d\) and \(i\) and of length 2 (and stretch 1) to \(j\). Thus, agent~\(u\)'s stretchcosts improve by \(Z - 1\).

Finally, if in the network shown in~\Cref{fig:12-irc}e agent~\(v\) changes its strategy $\{h,g,i\}$ to $\{e,f,u\}$ we get the network shown in~\Cref{fig:12-irc}b. With this strategy change, agent $v$ only loses greedy paths of length 1 to the three nodes \(g,h\) and \(i\) and of length 2 (and stretch 1) to \(j\) but it gains greedy paths of length 1 to the three nodes \(e,f\) and \(u\) and of length 2 (and stretch 1) to the two nodes \(c\) and \(d\). Thus, agent~\(v\)'s stretchcost improves by \(Z - 1\). As this is the same strategy-profile that we started with, we have found an improving response cycle.
\end{proof}

Interestingly, if the agents iteratively select best response strategies, then no such cyclic behavior can occur. Thus, best response dynamics starting from any initial network are guaranteed to converge to a NE network. However, we will show later that computing a best response is hard. 
\begin{restatable}{theorem}{thmtwoseven}
	\label{the:12-brc}
    In a 1-2-metric, best response cycles cannot exist.
\end{restatable}
\begin{proof}
	By \Cref{lem:12-1edges}, all best responses always build all 1-edges and thus, a best response cycle can never include changes of the 1-edges. By \Cref{lem:12-2edges-first}, no 2-edge that any agent builds can be part of a greedy path that is used by any other agent and thus, any best response cycle cannot consist solely of changes of 2-edges either. Thus, best response cycles cannot exist.	
\end{proof}
\subsubsection*{\textbf{Computational Complexity}}
Here, we investigate the computational complexity of computing a best response and of deciding if a given network is in NE. We use our characterization theorem and the tight connection to the \textsc{Minimum Dominating Set} problem, which asks for minimum size dominating set for a given network $G$. It is well-known that \textsc{Minimum Dominating Set} is NP-hard~\cite{garey1979computers}. We get the following dichotomy results:
\begin{restatable}{theorem}{thmtwoeight}
    \label{the:12-best-response-np}
    In a 1-2-metric, computing a best response and deciding if a given network is in NE is NP-hard for $\alpha > \frac12$ and polynomial time computable for $\alpha \leq \frac12$. 
\end{restatable}
\begin{proof}
     We first show that computing a best response strategy is NP-hard with $\alpha > \frac12$. For this, we reduce from the \textsc{Minimum Dominating Set} problem. 
     Given any graph \(G = (V, E)\) for which we want to compute a minimum size dominating set. We embed the nodes of $G$ as follows into a 1-2-metric: Let \(\mathcal{M} = (V \cup \{u\}, d)\) be a 1-2-metric, with $u\notin V$ being a new node, where for all \(x,y \in V \cup \{u\}\) we have
     \[d(x,y) = \begin{cases}
     1	&\text{, if }(x,y) \in E,\\
     2	&\text{, otherwise.}
     \end{cases}\]
     Furthermore, let \(\s\) be a strategy-profile where in $G(\s) = (V\cup \{u\},E(\s))$ all the edges in \(G\) are built in both directions but no other edges. We now consider the best response of agent $u$ and let \(D\) be the set of nodes that agent~\(u\) builds edges to in its best response in \(G(\s)\). Since we have $d(u,x) = 2$ for every node $x\in V$, it follows that node $u$ has a stretch of $1$ 
     to all nodes in $D$ and a stretch of~$\frac32$ to all nodes $w\in V\setminus D$ such that there exists a node $v\in D$ with $(v,w) \in E(\s)$, i.e, the edge $(v,w)$ is a 1-edge in $G(\s)$. Let $C$ be the set of these nodes, i.e., $C = \{w\in V\setminus D \mid \exists v\in D \wedge (v,w)\in E(\s)\}$. Note that for any node $z \in V \setminus (C\cup D\cup\{u\})$, there cannot exist a greedy path from $u$ to $z$, since, by \Cref{lem:12-stretches}, any such path can have at most two edges. Thus, since for any such node $z$ agent $u$ would incur stretchcost of $Z > \alpha + \frac32$, no such node $z$ can exist if $D$ is agent $u$'s best response. In this case, building an edge to $z$ would strictly decrease agent $u$'s cost, i.e., we have that $C\cup D = V$. Thus, the set $D$ must be a dominating set in $G$. Since every edge costs $\alpha > \frac12$ and building an edge to a node $x\in V$ can decrease its stretchcost at most by $\frac12$, agent $u$'s best response should buy as few edges as possible such that greedy routing is enabled, i.e., the set $D$ must be a minimum size dominating set in $G$. Thus, computing a best response strategy for agent $u$ is NP-hard.
     
     Next, if we have $\alpha \leq \frac12$ then the proof of \Cref{lem:12-dir-sum-const} implies for every agent $u$ that building all the edges that are incident outgoing edges of $u$ in the unique MaxDSG is a best response. This can be computed in polynomial time, since these are all outgoing edges, except the 2-edges $(u,w)$, where $(u,v)$ and $(v,w)$ both are 1-edges. 
    
     The NP-hardness of deciding if a given network $G$ is in NE in a 1-2-metric follows from \Cref{lem:12-dir-NE-BDSG}. There it is shown that for $\alpha > \frac{1}{2}$  all NE must be BDSG($
    \alpha$), i.e., the strategy of every agent $u$ must minimize the quantity $(\alpha-\frac{1}{2})|W_2(u)|+\frac{1}{2}|W_2^+(u)|$. When $|W_2^+(u)|=0$, this corresponds to a minimum  dominating set for the network $G_{-u}^1$ and this network can be arbitrary. 

     Finally, deciding if a given network is in NE with  $\alpha \leq \frac12$ can be done in polynomial time by \Cref{lem:12-dir-sum-const}.
\end{proof}

The characterization in \Cref{lem:12-dir-sum-const} and the above results directly imply that our results also hold for the problem of computing a NE network. This is true since MaxDSGs can be computed in polynomial time, whereas computing a MinDSG is NP-hard.
\begin{corollary}
In a 1-2-metric computing a NE can be done in polynomial time if $\alpha \leq \frac12$ and it is NP-hard if $\alpha> \frac12$.
\end{corollary}

\subsubsection*{\textbf{Greedy Equilibria}}
Since computing a best response is NP-hard if $\alpha > \frac{1}{2}$, it is natural to consider simpler strategy changes, where an agent can only add, delete or swap a single edge to decrease its cost. Networks where no such strategy changes can improve the agents' cost are exactly the set of Greedy Equilibria (GE)~\cite{lenznerGreedy2012}. For other related variants of network creation games, it has been shown that GE are good approximations of NE, i.e., that every GE is a 3-NE~\cite{lenznerGreedy2012,biloGeometric2019}. We now show, that enforcing greedy routing changes this picture completely.
\begin{restatable}{theorem}{thmtwoten}
	\label{the:12-ge}
    In a 1-2-metric there are GE, that are not in \(\Omega(\frac{\alpha n}{\alpha + n})\)-NE. 
\end{restatable}
\begin{proof}
We begin by showing that there are networks in GE that are not in \(\frac{\alpha n - 2 \alpha + n-1}{2\alpha + n}\)-approximate NE. For this, we examine the network $G$ in \Cref{fig:Greedy}.
\begin{figure}[h]
\includegraphics[width=0.5\textwidth]{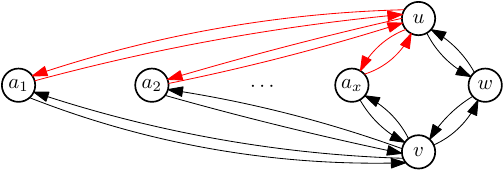}
\caption{A network in GE. Black edges represent 1-edges; red edges are 2-edges; all node pairs without edge inherit distance 2.}
\label{fig:Greedy}
\end{figure}

First, we show that network $G$ is in GE by proving that no agent can decrease its cost by adding, removing, or swapping a single incident outgoing edge. 

Let $\s$ be the strategy-profile such that $G = G(\s)$.
Agent \(u\) cannot remove its edges to any node \(a_i\) or \(w\) as that would remove the greedy path to that node. The stretch to node~\(v\) is already \(1\), so adding or swapping an edge would not be beneficial either.
No agent~\(a_i\) can remove its edge to nodes \(u\) or \(v\) as that would remove the greedy paths to them. The stretch to all other nodes~\(a_i\) and \(w\) is already \(1\), so adding or swapping an edge would not be beneficial either.
Agent~\(v\) cannot remove any edge to a node~\(a_i\) or \(w\) because that would remove the greedy paths to them. The stretch to node~\(u\) is already \(1\) so adding or swapping an edge would not be beneficial either.
Finally, agent~\(w\) cannot remove any edge to nodes~\(u\) or \(v\) because otherwise there would be no greedy path to that node. The stretches to all nodes~\(a_i\) are already \(1\), so adding an edge or swapping one would not be beneficial either. Thus, the network $G$ is in GE. 
Let $S_u = \{a_1,\cdots, a_x, w\}$, i.e., agent $u$'s strategy in $G(\s) = G(S_u,\s_{-u})$. The cost of agent $u$ is then given by
\[c_u(S_u,\s_{-u}) = \text{stretchcost}_u(S_u,\s_{-u}) + \text{edgecost}_u(S_u,\s_{-u}) = \alpha(x + 1) + n - 1 = n(\alpha + 1) - 2\alpha - 1,\]
since we have that $x = n - 3$.

Now, if in contrast agent~\(u\) would remove its edges to all nodes~\(a_i\) and instead build one edge to node~\(v\), i.e., the strategy $S_u' = \{v,w\}$, this yields a cost of
\[c_u(S_u',\s_{-u}) = \text{stretchcost}_u(S_u',\s_{-u}) + \text{edgecost}_u(S_u',\s_{-u}) = 2\alpha +\frac{3}{2}x + 2 = 2\alpha + \frac{3}{2}n - \frac{5}{2}.\]

Thus, the network is not in \(\frac{c_u(S_u,\s_{-u})}{c_u(S_u',\s_{-u})} = \frac{n(\alpha + 1) - 2\alpha - 1}{2\alpha + 3/2 n - 5/2} \in \Omega(\frac{\alpha n}{\alpha + n})\)-approximate NE.
\end{proof}

\subsubsection*{\textbf{Approximate Equilibria}} For $\alpha > \frac{1}{2}$, we know that NE exist, but even deciding if a given network is in NE is NP-hard. Thus, aiming for approximate equilibria seems an appropriate solution. However, we have seen that simply using GE for this is also an option. However, since approximation algorithms exist for the \textsc{Dominating Set} problem, there is a different approach. 

\begin{restatable}{theorem}{thmtwoeleven}
In 1-2-metrics, a $\mathcal{O}(\log n)$-approximate NE always exists.
\end{restatable}
\begin{proof}
By \Cref{lem:12-dir-NE-BDSG} we have that for $\alpha>0$ all $BDSG(\alpha)$ are the only NE. That means that for every node $u$ the quantity $b = \left(\alpha-\frac{1}{2}\right)|W_2(u)|+\frac{1}{2}|W_2^+(u)|$ should be minimum. Using the fact that $|W_2(u)|=|W_2^+(u)|+|W_2^-(u)|$, we get that $b = \left(\alpha-\frac{1}{2}\right)|W_2^-(u)|+\alpha|W_2^+(u)|$.

Then, we construct a network that is in $\mathcal{O}(\log n)$-approximate NE as follows. To ensure that greedy routing is enabled, we first construct a MaxDSG network $G$. Now, we consider every agent~$u$ sequentially. We know that agent $u$'s strategy must be a dominating set in the network $G^1_{-u}$. Since the \textsc{Dominating Set} problem is a special case of the \textsc{Set Cover} problem, we can use the standard greedy approximation algorithm for \textsc{Weighted Set Cover}~\cite{Chvatal79} to compute a dominating set that is at most a factor of $\mathcal{O}(\log n)$ larger than the minimum size dominating set, where every node $v\in W_2^+(u)$ has weight $\alpha$, and every node $x\in W_2^-(u)$ has weight $(\alpha-\frac{1}{2})$. We now replace agent $u$'s strategy in $G$ by the computed approximate dominating set. 

With this, we ensure that greedy routing is still enabled, and we have that no agent can decrease its edgecost by more than a factor of $\mathcal{O}(\log n)$ by any strategy change. Moreover, observe that, since greedy routing is enabled, we have, by \Cref{lem:12-stretches}, that all pairwise stretches are at either $1$ or $\frac32$. Thus, in any strategy that ensures that greedy routing works for some agent $u$, the stretchcosts of $u$ can be at most a factor of $\frac32 \in \mathcal{O}(\log n)$ higher that its best possible stretchcost. Thus, in the constructed network, any agent $u$ can improve its cost at most by a factor of $\mathcal{O}(\log n)$ by performing a strategy change.  
\end{proof}

\section{Tree Metrics}\label{sec:treemetrics}
As the next step we examine networks that are created with an underlying tree metric. In a tree metric, a positively weighted undirected spanning tree \(T\) is given, such that for all nodes \(u,v \in \points\), we have \(d(u,v) = d_T(u,v)\). 
In the following, we always use \(T\) to denote the given tree. 

Let $\s^T$ denote the strategy-profile, where every edge of~$T$ is created in both directions and let $G^T = G(\s^T)$ be the corresponding network.
For our analysis, we will consider the network $G^T = (\points,E^T,\ell)$ \emph{rooted at a node $r\in \points$}, denoted as $G^T_r$. This is defined analogous to rooting the tree~$T$ at node $r$: node \emph{$v$ is the child of node $u$ in $G^T_r$}, if $(u,v)\in E^T$ and if $d(u,r) < d(v,r)$, i.e., if $u$ is closer to $r$ than $v$. Moreover, in $G^T_r$ node $w$ is a \emph{descendant} of node $u$, if a path $u=x_1,x_2,\dots,x_k=w$ exists, such that $x_{i+1}$ is the child of $x_1$, for $1\leq i \leq k-1$. Also, $u$ is a descendant of itself.

Since for all our purposes the network $G^T_r$ behaves exactly like the tree $T$ rooted at $r$, we will from now on use the terminology from trees, when working with $G^T$ or $G^T_r$. For example, for $G^T_r$ we let \(subtree(u)\) denote the subtree of $G^T_r$ rooted at node $u$ that includes all descendants of \(u\) (including $u$). 
Furthermore, let \(below(u)\) refer to the set of subtrees \(\{subtree(v)\ |\ v\text{ is a child of }u\}\).
Using the above definitions, we get the following useful statements:
\begin{restatable}{lemma}{lemmathreeone}
	\label{lem:tree-subtree-paths}
    In a tree metric, a greedy path from node~\(u\) to a node~\(v\) can only consist of nodes that are in the same subtree from \(below(u)\) in $T$ rooted at $u$ that contains node~\(v\).
\end{restatable}
\begin{proof}
	Let \(u \in \mathcal{P}\) and let \(v \in \mathcal{P} \setminus \{u\}\). Now, consider the set of subtrees \(below(u)\) of $T$ rooted at $u$. Note that the subtrees in $below(u)$ include all agents except $u$. Hence, there exists a subtree~$T' \in below(u)$, such that $v \in T'$. To show our statement it thus suffices to show that, indeed, a path from $u$ to $v$ can only be a greedy path, if the path consists exclusively of nodes in $T'$. Towards a contradiction, consider any agent \(x \not\in T'\) and assume that $x$ is part of a greedy path from $u$ to $v$. Observe that agent~$x$ has a path to $v$ only via $u$, since $T$ is rooted at $u$ and $x$ and $v$ are nodes of two different subtrees of $u$. It follows that \(d_T(x,v) > d_T(u,v)\), which contradicts that the considered path is greedy path from $u$ to $v$ via $x$.
\end{proof}

\begin{restatable}{lemma}{lemmathreetwo}
\label{lem:tree-subtree-edges}
    In any strategy profile $\s$ that enables greedy routing with a tree metric, for any node $u$, in $G(\s)$ agent~\(u\) needs to have an edge to some node in every subtree in \(below(u)\) of $T$ rooted at $u$.
\end{restatable}
\begin{proof}
	Fix any agent \(u \in \mathcal{P}\) and consider the network $G(\s)$. For the sake of contradiction, let \(T' \in below(u)\) be a subtree of $T$, such that agent~\(u\) does not build an edge to any node~$v \in T'$ in $G(\s)$. 
    Then, fixing any $w \in T'$ and, by \Cref{lem:tree-subtree-paths}, all greedy paths from $u$ to $w$ can only traverse nodes in $T'$. Thus, by recalling that agent~$u$ has no edge to any vertex $v \in T'$, agent~\(u\) cannot have a greedy path to \(w\).
\end{proof}
\subsubsection*{\textbf{Equilibrium Existence}}
Here, we show the existence of equilibria and partially characterize them.

\begin{restatable}{theorem}{thmthreethree}
    \label{lem:tree-t}
    In a tree metric, the network $G^T$ is always a NE and a social optimum.
\end{restatable}
\begin{proof}
    Consider an agent \(u \in \mathcal{P}\). First, observe that in \(G^T\) all of agent~\(u\)'s distances and stretches are already minimal. Hence, adding any edges cannot improve its stretchcosts. Moreover, by \Cref{lem:tree-subtree-edges}, agent \(u\) also cannot remove any edges. For the same reason, swapping an edge $(u,v)$ with another edge $(u,w)$ is only possible for \(T' \in below(u)\), such that $u,v \in V(T')$, i.e., if the two nodes belong to the same subtree.  
    However, since swapping edges does not change agent $u$'s edgecosts, such a swap cannot be an improving move  since agent $u$'s stretchcosts cannot be further improved. Hence, there are no improving moves for any agent and \(G^T\) is a NE. Because all distances and stretches are minimal and there is no cheaper set of edges to enable greedy routing, the network \(G^T\) must also be social optimum.
\end{proof}

In the following, we show that GE are unique. This completely characterizes all GE and NE.

\begin{theorem}
    \label{lem:tree-dir-sum-pne}
    In a tree metric, the network \(G^T\) is the only GE.
\end{theorem}
\begin{proof}
    We prove the statement in two steps while assuming, for the sake of contradiction, that there is a GE network \(G\) that differs from \(G^T\). First, we show that in this case, network $G^T$ cannot be a proper subgraph of $G$, i.e., network $G$ contains all the edges of $G^T$ and at least one other edge. Then, in a second step, we show that if $G^T$ is not a subgraph of $G$, i.e., if $G$ does not contain all the edges of $G^T$, then $G$ is not in GE. Hence, any GE network is identical to $G^T$.
    
    To show that $G^T = (\points,E^T,\ell)$ cannot be a proper subgraph of $G = (\points, E, \ell)$, we observe that if this was the case, then $E^T \subset E$ holds, i.e., $E\setminus E^T \neq \emptyset$. Thus, there exists at least one edge $e \in E \setminus E^T$. However, by definition of $G^T$, the total strechcosts and the distances are minimized in the network~$G^T$. Hence, removing the edge $e \in E \setminus E^T$ would be an improving move. A contradiction.

    Let us now consider the case that $G^T$ is not a subgraph of $G$, while $G$ is a GE. We define $f_a(b)$ to be the number of descendants of node $b$ in $G^T_a$. We consider a specific edge, for this, let $(u,v) = \argmin_{(a,b) \in E^T\setminus E}f_a(b)$, i.e., the tuple $(u,v)$ that is an edge in $G^T$ but not in $G$, that  minimizes the number of descendants of node~$v$ in $G^T_u$. Notice that such an edge always exists, since we assume that $G^T$ is not a subgraph of $G$ and thus, $E^T\setminus E \neq \emptyset$. 
    
    We now examine the greedy path from node~$u$ to node~$v$ in network~$G$. Since $(u,v) \not \in E$, in conjunction with the assumption that $G$ is a GE, there must exist another node~$x \neq v$, such that $(u,x) \in E$, that enables a greedy path from $u$ to $v$ in $G$, i.e., the edge $(u,x)$ is the first edge on this greedy path. By \Cref{lem:tree-subtree-paths}, node~$x$ must belong to the same subtree $T'$ as $v$ in $below(u)$ in $T$ rooted at $u$. Since $v$ is a child of $u$ in $T$ rooted at $u$, we have that $T' = subtree(v)$. 
    It follows that $(u, x) \notin E^T$, since in $G^T_r$ node $u$ only has $v$ as child in $T'$. 

    Next, recall that we assume that $G = G(\s) = G(S_u,\s_{-u})$ is a GE. It follows that for agent~$u$ it is not an improving move to swap the edge~$(u,x) \in E$ for the edge~$(u,v) \notin E$. Now, let $S_u' = S_u \setminus\{x\}$ and let $\s^v = (S_u' \cup \{v\},\s_{-u})$. For convenience, let $\s^x = \s = (S_u' \cup \{x\},\s_{-u})$.

    We first show, that after the swap, i.e., in $G(\s^v)$, agent~$u$ still has a greedy path to all nodes. Since in $G(\s^v)$ nothing changes for greedy paths that do not use $x$ as first hop, it suffices to focus on such paths. Thus, consider a greedy path $P$ from $u$ to $z\in \points$ in $G(\s^x)$ that uses $x$ as first hop. Since $G$ is a GE, agent~$v$ has a greedy path $Q$ to $x$ in $G$. Moreover, since $x$ is in $subtree(v)$, it follows that $d_T(v,x) < d_T(u,x)$. Thus, in $G(\s^v)$ agent~$u$ has a greedy path to $x$ via node $v$ and $Q$. This implies that also in $G(\s^v)$ agent $u$ has a greedy path to $z$, since the path via $v$, $Q$ and $P$ is a greedy path.  

    The swap from $(u,x)$ to $(u,v)$ decreases agent $u$'s stretch to $v$ but it does not change $u$'s edgecost. Since the swap does not reduce agent $u$'s cost, there must be a node $w\in \points$ to which agent $u$'s stretch increases after the swap, i.e., $\text{stretch}_{G(\s^x)}(u,w) < \text{stretch}_{G(\s^v)}(u,w)$. It follows, that in $G = G(\s^x)$, agent $u$ has a greedy path via $x$ to $w$. 
    
    By \Cref{lem:tree-subtree-paths}, it follows that both nodes $x$ and $w$ must belong to the same subtree $T'$ in $below(u)$ in $T$ rooted at $u$, which must be $T' = subtree(v)$. 

    Now, if $\text{stretch}_G(v,w) = 1$, then agent $u$'s stretch to $w$ must be optimal in $G(\s^v)$ and hence, the stretch cannot have increased compared to $G(\s^x)$. Hence, we can assume that $\text{stretch}_G(v,w) > 1$. In this case, there there must be a vertex $p$ along the path from $v$ to $w$ in $T$ that in $G$ does not create the edge to the next node $q$ on that path. Thus, the edge $(p,q)$ is in $E^T \setminus E$. Now, note that $f_p(q) < f_u(v)$, since $q$'s subtree in $G^T_p$ is completely contained in $v$'s subtree in $G^T_u$. This is a contradiction to the choice of edge $(u,v)$, i.e., to $(u,v)$ having the minimum $f$-value. 

    Thus, every edge of $G^T$ must be contained in $G$ and thus, $G^T$ is the only GE.   
\end{proof}
\subsubsection*{\textbf{Dynamic Properties}} We investigate if NE networks in tree metrics can be found by iteratively selecting best responses. The answer is affirmative.

\begin{restatable}{theorem}{thmthreefive}
	\label{the:tree-acyclic}
    In a tree metric our game is weakly acyclic under best responses.
\end{restatable}
\begin{proof}
    Fix any $u \in \mathcal{P}$ and root $T$ in $u$. Consider any subtree \(T' \in below(u)\). Recall that, by \Cref{lem:tree-subtree-edges}, agent $u$ needs to build an edge $(u,v)$, where $v \in T'$, in order to enable greedy routing from $u$ to the set of nodes in $T'$. Note that in case that in a strategy-profile $\s$ all edges of $T'$ are created in both directions, then the best response of agent~$u$ in $\s$ is to create the edge $(u, u')$, where $u'$ is the root of subtree $T'$, since $u' = \argmin_{v \in V(T')}d_T(u,v)$.

    We can now use this observation to get from any strategy-profile~$\s$ to the strategy-profile~\(\s^T\) by a finite sequence of best responses. For this, we root \(T\) in \(u\) and activate the agents of $T$ in a bottom-up fashion, i.e., starting with the leaves and then moving upwards.
    This ensures that all edges of \(G^T\) are eventually created.

    After all edges of $G^T$ are created, we activate all agents that still have edges that do not belong to $G^T$ in any order. Any such edge will be removed since the edges of $G^T$ already suffice to achieve optimal stretchcosts. Thus, since $G^T$ is a NE, our game is weakly acyclic under best responses. 
\end{proof}

\subsubsection*{\textbf{Computational Complexity}}
Here, we investigate the computational complexity of computing a best response in tree metrics. We show, that this is NP-hard. However, since by \Cref{lem:tree-dir-sum-pne} $G^T$ is the unique GE, deciding NE existence and computing a NE is tractable.

\begin{restatable}{theorem}{thmthreesix}
\label{the:tree-np}
	In tree metrics, computing a best response is NP-hard.
\end{restatable}
\begin{proof}
The main argument of the proof is  that computing a best response boils down to solving \textsc{Set Cover} \cite{karpReducibility1972}.
 To this end, we reduce from  \textsc{Set Cover} and show that if the number of edges build in a best response strategy of an agent in our constructed instance could be computed in polynomial time, then the size of a minimum set cover could also be computed in polynomial time.

\begin{figure}[htb]
\begin{subfigure}[c]{0.5\textwidth}
    \centering
    \resizebox{0.4\linewidth}{!}{
    \begin{tikzpicture}[bend angle=60,
        help/.style={circle, minimum size=20mm}],
    \node[vertex]   (c)												{\(c\)};
    \node[vertex]   (u)			[left=1.4cm of c]						{\(u\)}
        edge []		node[auto,font=\LARGE] {2} 				(c);
    \node[vertex]   (v)			[right=1.4cm of c]					{\(v\)}
        edge []		node[auto,swap,font=\LARGE] {1} 				(c);
        
    \node[vertex]   (Q1)        [above left=1.4cm and 1cm of c]		{\(Q_1\)}
        edge []		node[auto,font=\LARGE] {1} 				(c);
    \node[fill=white]   (Se)    [above=2cm of c]					{...}
        edge []		node[auto,font=\LARGE] {1} 				(c);
    \node[vertex]   (Qm)        [above right=1.4cm and 1cm of c]		{\(Q_m\)}
        edge []		node[auto, font=\LARGE] {1} 				(c);    
        
    \node[vertex]   (x1)        [below left=1.4cm and 1cm of c]		{\(x_1\)}
        edge []		node[auto,font=\LARGE] {2} 				(c);
    \node[fill=white]   (xe)    [below=1.5cm of c]					{...}
        edge []		node[auto,swap,font=\LARGE] {2} 				(c);
    \node[vertex]   (xn)        [below right=1.4cm and 1cm of c]		{\(x_n\)}
        edge []		node[auto,swap,font=\LARGE] {2} 				(c);
        
    \end{tikzpicture}}
        	   \caption{} 
	    \end{subfigure}
         \hfill
 \begin{subfigure}[c]{0.49\textwidth}
    \centering
    \resizebox{0.6\linewidth}{!}{
    \begin{tikzpicture}[bend angle=60,
        help/.style={circle, minimum size=20mm}],
    \node[fill=white]   (c)												{};
    \node[vertex]   (v)			[right=3.4cm of c]					{\(v\)};

    \node[vertex]   (u)			[right=1cm of v]						{\(u\)};
    \node[vertex]   (c1)			[above=0.8cm of v]					{\(c\)}
             edge [<-]		node[auto, font=\LARGE] {} 				(v);

    \node[vertex]   (x1)        [above left=1.4cm and 0.8cm of c]		{\(x_1\)};
    \node[fill=white]   (Se)    [right=1cm of c]					{\dots};

    \node[vertex]   (Q1)        [above right=1.4cm and 1cm of c]		{\(Q_1\)}
        edge [->]		node[auto, font=\LARGE] {} 				(x1)   
        edge [<-]		node[auto, font=\LARGE] {} 				(v);

    \node[vertex]   (xn)        [below left=1.4cm and 0.8cm of c]		{\(x_n\)}
        edge [<-]		node[auto,font=\LARGE] {} 				(Q1);
    \node[fill=white]   (xe)    [left=0.8cm of c]					{\dots};
    \node[vertex]   (Qm)        [below right=1.4cm and 1cm of c]		{\(Q_m\)}
        edge [->]		node[auto,swap,font=\LARGE] {} 				(xn)
        edge [<-]		node[auto, font=\LARGE] {} 				(v); 
        
    \end{tikzpicture}}
            \caption{}	    
	    \end{subfigure}
    \caption{(a) tree $T$ that defines the tree metric; (b) our construction of the \textsc{Set Cover}-instance.}
    \label{fig:NPTree}
\end{figure}
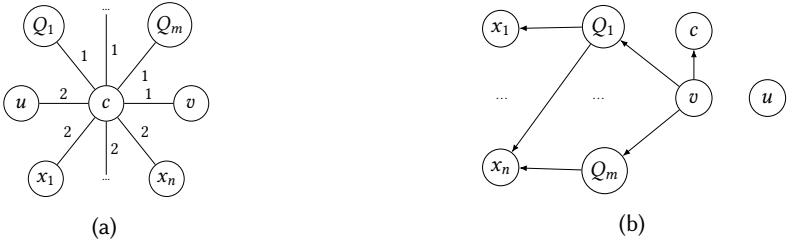

Let \(P = (\{x_1, ... x_n\},\{Q_1, ..., Q_m\})\) be an instance of \textsc{Set Cover}, where $\{x_1, ... x_n\}$ is the set of elements that need to be covered and $\{Q_1, ..., Q_m\}$ is a collection of subsets of the elements that can be selected. 
Given a \textsc{Set Cover}-instance, we construct a corresponding bipartite graph $G(V,E)$ with a node for every set \(Q_i\) and for every element \(x_i\), and edges between every set and its contained elements. We add nodes \(u,v\) and \(c\) and connect node~\(v\) to node~\(c\) as well as every node \(Q_i\). The tree metric then is defined as in \Cref{fig:NPTree}a and we fix \(\alpha > 4\). 

Let \(C\) be a minimum set cover of \(P\). We show that the number of edges that agent~\(u\) builds in a best response in the corresponding game instance is \(1 + |C|\). First, observe that node~$v$ only inherits outgoing edges. Thus, agent~\(u\) has to create an edge to 
node~\(v\) to enable a greedy path to \(v\). In contrast, consider the potential edge $(u,c)$. Note that $\frac{d(u,v)+d(v,c)}{d(u,c)} = 2 < 4 < \alpha$. This, in conjunction with the observation that a potential edge $(u,c)$ would not affect the stretch to any other node, leads to the conclusion that the creation of an edge from $u$ to $c$ is not an improving move.

Let $B$ be the best response of agent $u$. Given $B$, where we do not necessarily know which nodes it contains, except for $v\in B$ and $c \not\in B$, we construct another strategy~$B' = \psi(B)$. The mapping $\psi: X \rightarrow Y$, for $X,Y \subseteq V$, is defined as follows: For $Q_j \in B$ or $v \in B$ we simply map $\psi(Q_j) = Q_j$ and $\psi(v) = v$ respectively. In contrast, for any $x_i \in B$ we have $\psi(x_i) = Q_j$, where $(Q_j, x_i) \in E(G)$. Note that for any $x_i$ such a node $Q_j$ exists by the construction of $G$. Moreover, we observe that whenever $x_i \in B$, then this implies that $Q_j \not\in B$. To see this, consider any such pair $(Q_j, x_i)$ and assume that $Q_j, x_i \in B$, i.e., in agent~$u$'s best response it builds an edge to both $x_i$ and $Q_j$. 

We show that dropping the edge $(u, x_i)$ would be a better response, hence, a contradiction to $B$ being the best response of agent~$u$. To that end, consider the stretchcosts without $(u, x_i)$. As $u$ has a greedy path to $x_i$ via $Q_j$, we get $\text{stretch}_G(u,x_i) = \frac{d(u, Q_j) + d(Q_j, x_i)}{d(u, x_i)} = \frac{3}{2}$ by our metric in \Cref{fig:NPTree}~a. 

In contrast, including the edge $(u, x_i)$ entails an additional cost of $\alpha$ with optimal stretch $\text{stretch}_G(u,x_i) = 1$. Since $\alpha > 4$, this yields a larger cost for $u$ compared to dropping the edge. The contradiction is completed by noticing that the edge $(u, x_i)$ does not enable any other greedy path for $u$. We conclude that our defined mapping $\psi: X \rightarrow Y$ is valid for our set $B$ and that $|B| = |B'|$ for $B' = \psi(B)$.

Next, we observe that $B' \setminus \{v\}$ constitutes a set cover $P$. To see that, recall that greedy routing is enabled for agent $u$ via strategy $B$. Thus, strategy $B'$ also enables greedy routing for agent $u$ which entails that for all $x_i$ there exists an $Q_j$ such that $(u, Q_j)$ and $(Q_j, x_i)$. 

Moreover, the set \(B' \setminus \{v\}\) must constitute a minimum set cover: Suppose that it were otherwise and a set cover \(C\) with \(|C| < |B'| - 1\) exists. We show that \(D = C \cup \{v\}\) is a better response than \(B\), which yields a contradiction, since $B$ is the best response for agent $u$. In both strategies, agent \(u\)'s distances to \(c\) and \(v\) are \(d(u,c) = 4\)  and \(d(u,v) = 3\)  respectively, the distances to all \(Q_j\) are either \(d(u,Q_j) = 3\) or \(d(u,v) + d(v,Q_j) = 5\) depending on whether there exists an edge $(u, Q_j)$. Thus, the difference in stretch of a network including $(u, Q_j)$, in contrast to the same network but without edge~$(u, Q_j)$, is $\Delta = \frac{d(u,v) + d(v, Q_j)}{d(u, Q_j)} - \frac{d(u, Q_j)}{d(u,Q_j)} = \frac{2}{3}$. In a similar fashion, recalling that if $(u,x_i)$ is not included then agent $u$ creates the edge~$(u, Q_j)$, we have that the distance to \(x_i\) is either \(d(u,x_i) = 4\) or \(d(u,Q_j) + d(Q_j,x_i) = 6\) depending on whether there is an edge $(u, x_i)$. Moreover, the difference in stretch is then given by $\delta = \frac{d(u,Q_j) + d(Q_j, x_i)}{d(u,x_i)} - \frac{d(x_i)}{d(x_i)} = \frac{1}{2}$. Then, setting $X = \{x_1, \cdots, x_n\}$ and $Q = \{Q_1, \cdots, Q_m\}$, the stretchcosts of \(B\) are at least
\begin{align*}
    &\frac{d(u,v)+d(v,c)}{d(u,c)} + \frac{d(u,v)}{d(u,v)} + \sum_{Q_j \in Q \setminus B}^m\frac{d(u,v)+d(v,Q_j)}{d(u,Q_j)} + \sum_{x_i \in X \setminus B}^n\frac{d(u,Q_{i,j})+d(Q_{i,j},x_i)}{d(u,x_i)} + |B| - 1\\
    &\geq \frac{d(u,v)+d(v,c)}{d(u,c)} + \frac{d(u,v)}{d(u,v)} + \sum_{j=1}^m\frac{d(u,v)+d(v,Q_j)}{d(u,Q_j)} + \sum_{i=1}^n\frac{d(u,Q_{i,j})+d(Q_{i,j},x_i)}{d(u,x_i)} - max(\delta, \Delta)(|B|-1),   
\end{align*}
where $Q_{i,j}$ is $Q_j$ such that $(Q_j, x_i) \in E(G)$. The last line follows as we overcount the number of edges by a factor of $|B| - 1$ and we assume the discount factor of the stretches to be maximal $\Delta = max(\delta, \Delta)$ in order to achieve a lower bound. Plugging in the values now yields
\begin{equation}\label{eq:lowerbound-stretch}
    \text{stretchcost}_u(B) \geq \frac{4}{2} + \frac{3}{3} + \frac{5}{3}m + \frac{6}{4}n - \frac{2}{3}(|B|-1).
\end{equation}

Analogously, we derive an upper bound on the stretchcost of agent $u$ applying strategy $D$, giving
\begin{align}\label{eq:upperbound-stretch}
\begin{split}
    &\frac{d(u,v)+d(v,c)}{d(u,c)} + \frac{d(u,v)}{d(u,v)} + \sum_{Q_j \in Q \setminus D}^m\frac{d(u,v)+d(v,Q_j)}{d(u,Q_j)} + \sum_{x_i \in X \setminus D}^n\frac{d(u,Q_{i,j})+d(Q_{i,j},x_i)}{d(u,x_i)} + |D| - 1\\
    &\leq \frac{d(u,v)+d(v,c)}{d(u,c)} + \frac{d(u,v)}{d(u,v)} + \sum_{j=1}^m\frac{d(u,v)+d(v,Q_j)}{d(u,Q_j)} + \sum_{i=1}^n\frac{d(u,Q_{i,j})+d(Q_{i,j},x_i)}{d(u,x_i)} - min(\delta, \Delta)(|D|-1)\\
    &= \frac{4}{2} + \frac{3}{3} + \frac{5}{3}m + \frac{6}{4}n - \frac{1}{2}(|D|-1).
\end{split} 
\end{align}

Thus, combining \Cref{eq:lowerbound-stretch} with \Cref{eq:upperbound-stretch}, the increase in stretch cost from strategy $B$ in comparison to strategy $D$ is at most 
\begin{align*}
    \text{stretchcost}_u(D) - \text{stretchcost}_u(B) \leq \frac{2}{3}|B| - \frac{1}{2}|D| - \frac{1}{6} < \frac{2}{3}|B| - \frac{1}{2}|D|.
\end{align*}

This now stays in contrast to the reduction in edgecosts given by $\text{edgecost}_u(B) - \text{edgecost}_u(D) = (|B| - |D|)\alpha$. We now recall that $\alpha > 4$ and by previous discussion it is finally revealed that 
\begin{align*}
    \text{edgecost}_u(B) - \text{edgecost}_u(D) &= (|B| - |D|)\alpha\\ 
    &> (|B| - |D|)4 \\
    &> \frac{2}{3}|B| - \frac{1}{2}|D|\\
    &> \text{stretchcost}_u(D) - \text{stretchcost}_u(B),
\end{align*}
whereby the penultimate line follows since $|D| \leq |B| - 1$. Notice that $\text{edgecost}_u(B) - \text{edgecost}_u(D) > \text{stretchcost}_u(D) - \text{stretchcost}_u(B)$ entails that $c_u(B) > c_u(D)$. Hence, if a best response of size $k$ could be computed in polynomial time, then the size of a set cover of size $k-1$ could be found in polynomial time, since our reduction is computable in polynomial time as well.
\end{proof}
\section{Euclidean Metrics}\label{sec:eumetrics}
We study Euclidean metrics, which are metrics where there is a function that maps agents to points in a \(d\)-dimensional Euclidean space, such that the distances in the metric between agents correspond to the distances of their points in the Euclidean space, measured via the 2-norm. We focus on 2D-Euclidean metrics but all results regarding the existence of equilibria and computational complexity directly apply to higher dimensional spaces as well. Given three nodes \(u,v,w \in \points\) in a Euclidean metric, let \(\angle vuw\) be the angle of \(u\) formed by rays \(\overrightarrow{uv}\) and \(\overrightarrow{uw}\). We always consider the positive angle, which is at most~\(\pi\).

\subsubsection*{\textbf{Equilibrium Existence}}
First, we show that GE may not exist, modifying the proof of Theorem 5.1 from Moscibroda, Schmid and Wattenhofer~\cite{moscibrodaTopologies2006}.

\begin{restatable}{theorem}{thmfourone}
    \label{the:existenceDirectedEuclid}
    In a 2D Euclidean metric there are instances of our game that do not have a GE.
\end{restatable}
\begin{proof}
\begin{figure}[htb]
    \centering
	    \centering
		    \resizebox{0.35\linewidth}{!}{
	    \begin{tikzpicture}[bend angle=90],
	    \clip(-1.25cm,-5.75cm) rectangle (8.5cm,2.25cm);
	    \node[vertex]   (a)                                     {\(a\)};
	    \node[vertex]   (b)        [right=3cm of a.center]                 {\(b\)}
	        edge []                          node[auto,swap, font=\LARGE] {1.14} (a);
	    \node[vertex]   (c)        [right=3cm of b.center]                 {\(c\)}
	        edge [bend right=45]			 node[auto,swap, font=\LARGE] {\(\approx 2.14\)} (a)
	        edge []                          node[auto,swap, font=\LARGE] {1} (b);
	        
	    \node[vertex]   (y)        [below left=3cm and 1.5cm of b.center]                 {\(y\)}
	        edge []                          node[auto, font=\LARGE] {1.96} (a)
	        edge []                          node[auto, font=\LARGE] {\(2-\epsilon\)} (b)
	        edge [looseness=2,bend right]   node[auto, font=\LARGE] {\(\approx 2.45\)} (c);
	    \node[vertex]   (z)        [right=3cm of y.center]                 {\(z\)}
	        edge [looseness=2,bend left]    node[auto,swap, font=\LARGE] {\(\approx 2.47\)} (a)
	        edge []                          node[auto,swap, font=\LARGE] {2} (b)
	        edge []                          node[auto,swap, font=\LARGE] {\(2+\epsilon\)} (c)
	        edge []                          node[auto,swap, font=\LARGE] {\(1-2\epsilon\)} (y);
	    
	    \end{tikzpicture}}
	    \caption{2D Euclidean metric that does not have a GE.}
     	    \label{fig:existenceDirectedEuclid}

\end{figure}
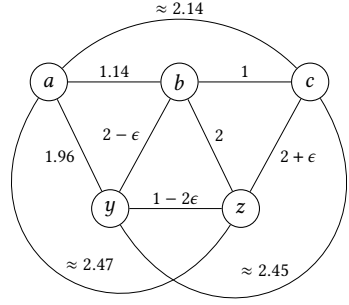
Modifying the proof of Theorem 5.1 from Moscibroda, Schmid and Wattenhofer~\cite{moscibrodaTopologies2006} for \(k = 1\), where \(\epsilon>0\) is an arbitrarily small constant, we show that there are no GE in the network in \Cref{fig:existenceDirectedEuclid} for \(\alpha = 0.6\). Differing from the construction in~\cite{moscibrodaTopologies2006}, we changed \(d(b,y)\) to enable greedy routing. To that end, we first establish general statements about the strategies of agents in NE to show that only the strategies for agents \(y\) and \(z\) shown in \Cref{fig:candidateEuclid} could possibly be part of a NE. Then, we show that none of these strategies is in NE. 

First, we note that the edges \((b,a),(b,c),(c,b),(y,z)\) and \((z,y)\) have to be built to enable greedy routing, because there is no other node closer to the head of each edge than the tail of that edge. Furthermore, agent \(a\) also builds an edge to \(b\) because the only other possible greedy path via \(c\) would have a stretch of \(\frac{d(a,c)+d(c,b)}{d(a,b)} = \frac{d(a,c)+1}{1.14} > 1 + \alpha\) and thus adding the edge to \(b\) is an improving move for \(a\). Additionally, agent \(a\) does not build an edge to \(c\) because that would improve the stretch to \(c\) by \(\frac{d(a,b) + d(b,c)}{d(a,c)} - 1 = \frac{2.14}{d(a,c)} - 1 < \alpha\) without changing the stretch to any other agent. Analogously, agent \(c\) does not build an edge to \(a\) because that would decrease stretches by at most \(\frac{d(c,b) + d(b,a)}{d(c,a)} - 1 + \frac{d(c,b) + d(b,a) + d(a,y) - (d(c,a) + d(a,y))}{d(c,y)} = \frac{2.14}{d(a,c)} - 1 + \frac{4.1 - (d(a,c) + 1.96)}{d(c,y)} < \alpha\).

Moreover, agent \(y\) also builds an edge to \(a\) because the only other possible greedy path via \(b\) has a stretch of \(\frac{d(y,b) + d(b,a)}{d(y,a)} = \frac{3.14 - \epsilon}{1.96} > 1 + \alpha\). Analogously, agent \(y\) does not build edges to both \(b\) and \(c\) because if it would, it could simply remove the edge to \(c\) and increase stretches by \(\frac{d(y,b) + d(b,c)}{d(y,c)} = \frac{3-\epsilon}{d(c,y)} - 1 < \alpha\). In fact, agent \(y\) does not build an edge to \(c\) because if they would, they could swap their edge from \(c\) to \(b\), changing the sum of stretches by at most \(\frac{d(y,b) + d(b,c)}{d(y,c)} - \frac{d(y,c) + d(c,b)}{d(y,b)} = \frac{3-\epsilon}{d(c,y)} - \frac{d(c,y) + 1}{2-\epsilon} < 0\).

Additionally, agent \(z\) has to build an edge to \(b\) or \(c\) because otherwise, there is no greedy path to \(c\). Agent \(z\) does not build both edges because if it did, it could remove its edge to \(c\), increasing stretches by \(\frac{d(z,b) + d(b,c)}{d(z,c)} - 1  = \frac{3}{2+\epsilon} - 1 < \alpha\). Finally, agent \(z\) does also not build an edge to \(a\) because the stretch to \(a\) via \(y\) is already \(\frac{d(z,y) + d(y,a)}{d(z,a)} = \frac{2.96-2\epsilon}{d(z,a)} < 1 + \alpha\) and the edge to \(a\) would not provide shorter greedy paths to \(b\) or \(c\)  because \(z\) needs an edge to \(b\) or \(c\) in any case.

This leaves the potential strategies for nodes \(y\) and \(z\) shown in \Cref{fig:candidateEuclid}. The strategies of \(a, b\) and \(c\) beyond the edges already discussed do not matter for the strategies for \(y\) and \(z\) because any additional edges could not be part of a greedy path starting from them. We examine all of these possible cases and show that they cannot be part of a GE.

\begin{figure}
    \centering
    \begin{subfigure}[b]{0.24\textwidth}
         \centering
         
        \resizebox{0.7\linewidth}{!}{
        \begin{tikzpicture}[bend angle=90],
        \node[vertex]   (a)                                     {\(a\)};
        \node[vertex]   (b)        [right=1cm of a.center]                 {\(b\)}
            edge [<->]                          (a);
        \node[vertex]   (c)        [right=1cm of b.center]                 {\(c\)}
            edge [<->]                          (b);
            
        \node[vertex]   (1)        [below left=1cm and 0.5cm of b.center]                 {\(y\)}
            edge [->]                          (a);
        \node[vertex]   (2)        [right=1cm of 1.center]                 {\(z\)}
            edge [->]                          (b)
            edge [<->]                          (1);
        \end{tikzpicture}}
        \caption{}
        \label{fig:candidateaEuclid}
    \end{subfigure}
    \hfill
    \begin{subfigure}[b]{0.24\textwidth}
         \centering
        \resizebox{0.7\linewidth}{!}{
        \begin{tikzpicture}[bend angle=90],
        \node[vertex]   (a)                                     {\(a\)};
        \node[vertex]   (b)        [right=1cm of a.center]                 {\(b\)}
            edge [<->]                          (a);
        \node[vertex]   (c)        [right=1cm of b.center]                 {\(c\)}
            edge [<->]                          (b);
            
        \node[vertex]   (1)        [below left=1cm and 0.5cm of b.center]                 {\(y\)}
            edge [->]                          (a)
            edge [->]                          (b);
        \node[vertex]   (2)        [right=1cm of 1.center]                 {\(z\)}
            edge [->]                          (b)
            edge [<->]                          (1);
        \end{tikzpicture}}
        \caption{}
        \label{fig:candidatecEuclid}
    \end{subfigure}
    \hfill
    \begin{subfigure}[b]{0.24\textwidth}
         \centering
        \resizebox{0.7\linewidth}{!}{
        \begin{tikzpicture}[bend angle=90],
        \node[vertex]   (a)                                     {\(a\)};
        \node[vertex]   (b)        [right=1cm of a.center]                 {\(b\)}
            edge [<->]                          (a);
        \node[vertex]   (c)        [right=1cm of b.center]                 {\(c\)}
            edge [<->]                          (b);
            
        \node[vertex]   (1)        [below left=1cm and 0.5cm of b.center]                 {\(y\)}
            edge [->]                          (a)
            edge [->]                          (b);
        \node[vertex]   (2)        [right=1cm of 1.center]                 {\(z\)}
            edge [->]                          (c)
            edge [<->]                          (1);
        \end{tikzpicture}}
        \caption{}
        \label{fig:candidatedEuclid}
    \end{subfigure}
    \hfill
    \begin{subfigure}[b]{0.24\textwidth}
         \centering
        \resizebox{0.7\linewidth}{!}{
        \begin{tikzpicture}[bend angle=90],
        \node[vertex]   (a)                                     {\(a\)};
        \node[vertex]   (b)        [right=1cm of a.center]                 {\(b\)}
            edge [<->]                          (a);
        \node[vertex]   (c)        [right=1cm of b.center]                 {\(c\)}
            edge [<->]                          (b);
            
        \node[vertex]   (1)        [below left=1cm and 0.5cm of b.center]                 {\(y\)}
            edge [->]                          (a);
        \node[vertex]   (2)        [right=1cm of 1.center]                 {\(z\)}
            edge [->]                          (c)
            edge [<->]                          (1);
        \end{tikzpicture}}
        \caption{}
        \label{fig:candidatebEuclid}
    \end{subfigure}
    \caption{Candidates for equilibrium strategies for vertices \(y\) and \(z\).}
    \label{fig:candidateEuclid}
\end{figure}
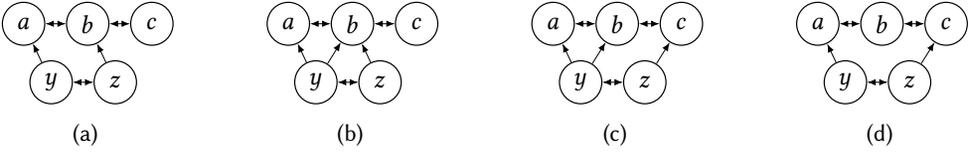

\textbf{Case 1:} In \Cref{fig:candidateaEuclid}, if \(y\) adds an edge to \(b\), the stretch to \(b\) decreases from \(\frac{d(y,a) + d(a,b)}{d(y,b)} = \frac{3.1}{2 - \epsilon}\) to \(1\) and the stretch to \(c\) from \(\frac{d(y,z) + d(z,b) + d(b,c)}{d(y,c)} = \frac{4-2\epsilon}{d(y,c)}\) to \(\frac{d(y,b) + d(b,c)}{d(y,c)} = \frac{3-\epsilon}{d(y,c)}\) which is a total improvement of more than \(\alpha\).

\textbf{Case 2:} In \Cref{fig:candidatecEuclid}, agent \(z\) can swap its edge from \(b\) to \(c\), changing the sum of stretches by \(\frac{d(z,y) + d(y,b)}{d(z,b)} - \frac{d(z,b) + d(b,c)}{d(z,c)} = \frac{3-3\epsilon}{2} - \frac{3}{2+\epsilon} < 0\).

\textbf{Case 3:} In \Cref{fig:candidatedEuclid}, if agent \(y\) removes its edge to \(b\), the stretch to \(c\) stays unchanged at \(\frac{d(y,z) + d(z,c)}{d(y,c)} = \frac{d(y,b) + d(b,c)}{d(y,c)}\), while the stretch to \(b\) increases from \(1\) to \(\frac{d(y,a) + d(a,b)}{d(y,b)} = \frac{3.1}{2 - \epsilon}\) which is an increase of less than \(\alpha\).

\textbf{Case 4:} In \Cref{fig:candidatebEuclid}, agent \(z\) can swap their edge from \(c\) to \(b\), changing the sum of stretches by \(\frac{d(z,b) + d(b,c)}{d(z,c)} - \frac{d(z,c) + d(c,b)}{d(z,b)} = \frac{3}{2 + \epsilon} - \frac{3+\epsilon}{2} < 0\).

Thus, there can be no GE in this instance. 
\end{proof}
\subsubsection*{\textbf{Computational Complexity}}
Next, we show that finding a best response is computationally hard.

\begin{restatable}{theorem}{thmfourtwo}
\label{the:euclid-np}
	In a 2D Euclidean metric, computing a best response is NP-hard.
\end{restatable}
\begin{proof}
We reduce from \textsc{Set Cover} and show that, if the number of edges of an agent's best response could be computed in polynomial time, then the size of a minimum set cover could also be found in polynomial time, which would imply P = NP.

\begin{figure}[htb]
\begin{subfigure}[c]{0.49\textwidth}
    \centering
    \resizebox{0.7\linewidth}{!}{
    \begin{tikzpicture}[bend angle=45,
        help/.style={circle, draw=gray, minimum size=20mm}],
    \node[vertex]   (u)										{\(u\)};
    \node[vertex]   (v)			[right=2cm of u]       		{\(v\)}
        edge []			node[auto,swap, font=\LARGE] {6} 				(u);
        
    \node[help]   (Scenter)        [above right=2cm and 1cm of v]      {}
        edge []			node[auto,swap, font=\LARGE] {9}         (u)
        edge []			node[auto,swap, font=\LARGE] {7} 		(v);
        
    \node[help]   (xcenter)        [below right=0.8 cm and 2.5cm of v]      {}
        edge []			node[auto, font=\LARGE]		{17}         (u)
        edge []			node[auto,swap,pos=0.7, font=\LARGE] {\(\approx 11.10\)} 		(v)
        edge []			node[auto,swap, font=\LARGE] {12}         (Scenter);
        
    \node[vertex]   (Qm)        at (Scenter.180)     {\(Q_m\)};
    \node[fill=white]   (Qe)    at (Scenter.270)     {\dots};
    \node[vertex]   (Q2)        at (Scenter.0)     {\(Q_2\)};
    \node[vertex]   (Q1)        at (Scenter.90)     {\(Q_1\)}
        edge [gray]			node[auto,swap, font=\LARGE] {\(\epsilon\)}         (Qe);
        
    \node[vertex]   (xm)        at (xcenter.180)     {\(x_n\)};
    \node[fill=white]   (xe)    at (xcenter.270)     {\dots};
    \node[vertex]   (x2)        at (xcenter.0)     {\(x_2\)};
    \node[vertex]   (x1)        at (xcenter.90)     {\(x_1\)}
        edge [gray]			node[auto,swap, font=\LARGE] {\(\epsilon\)}         (xe);
    
    \end{tikzpicture}}
    \caption{}
    \end{subfigure}
     \begin{subfigure}[c]{0.5\textwidth}
    \centering
    
    \resizebox{0.7\linewidth}{!}{
    \begin{tikzpicture}[bend angle=60,
        help/.style={circle, minimum size=20mm}],
    \node[fill=white]   (c)												{};
    \node[vertex]   (v)			[right=3.4cm of c]					{\(v\)};

    \node[vertex]   (u)			[right=1cm of v]						{\(u\)};

    \node[vertex]   (x1)        [above left=1.4cm and 0.8cm of c]		{\(x_1\)};
    \node[fill=white]   (Se)    [right=1cm of c]					{\dots};

    \node[vertex]   (Q1)        [above right=1.4cm and 1cm of c]		{\(Q_1\)}
        edge [->]		node[auto, font=\LARGE] {} 				(x1) 
        edge [<-]		node[auto, font=\LARGE] {} 				(v);

    \node[vertex]   (xn)        [below left=1.4cm and 0.8cm of c]		{\(x_n\)}
        edge [<-]		node[auto,font=\LARGE] {} 				(Q1);
    \node[fill=white]   (xe)    [left=0.8cm of c]					{\dots};
    \node[vertex]   (Qm)        [below right=1.4cm and 1cm of c]		{\(Q_m\)}
        edge [->]		node[auto,swap,font=\LARGE] {} 				(xn)
        edge [<-]		node[auto, font=\LARGE] {} 				(v); 
    \node[fill=white]   (zero1)	[above= 1cm of Q1] {};
    \node[fill=white]   (zero2)	[below= 1cm of Qm] {};
    \end{tikzpicture}}
            \caption{}	    
	    \end{subfigure}
    \caption{(a) 2D Euclidean metric where \(Q_i\) and \(x_i\) are equally spaced on a circle with diameter \(\epsilon>0\). Distances to these circles are the distance to the closest point on it, (b) our construction for the  \textsc{Set Cover} reduction. }
    \label{fig:NPEuclidean}
\end{figure}
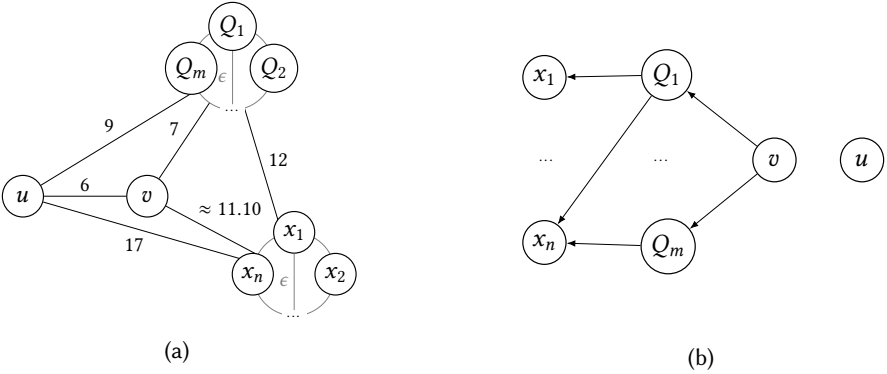

Given an instance of \textsc{Set Cover} \(P = (\{x_1, ... x_n\},\{Q_1, ..., Q_m\})\), we construct a bipartite graph analogously to the proof of \Cref{the:tree-np}. Then, by using the construction as given in \Cref{fig:NPEuclidean} and setting \(\alpha > 4\), the reduction from \textsc{Set Cover} follows by an almost verbatim repeating
of the proof of \Cref{the:tree-np}.

We use the metric as given in \Cref{fig:NPEuclidean}~a. 
Moreover, for the minimal and maximal stretch gain we have $\delta = \frac{d(u,Q_j) + d(Q_j, x_i)}{d(u,x_i)} - \frac{d(x_i)}{d(x_i)} = \frac{4}{17}$ and $\Delta = \frac{d(u,v) + d(v, Q_j)}{d(u, Q_j)} - \frac{d(u, Q_j)}{d(u,Q_j)} = \frac{4}{9}$ respectively.

This in conjunction with similar calculations as carried out in \Cref{eq:lowerbound-stretch} and \Cref{eq:upperbound-stretch} gives for the best response $B$ and the potential minimal set cover $D$  

\begin{align*}
    \text{stretchcost}_u(D) - \text{stretchcost}_u(B) &< \frac{4}{9}|B| - \frac{4}{17}|D|\\
    &<(|B| - |D|)\alpha =\text{edgecost}_u(B) - \text{edgecost}_u(D),
\end{align*}
where we made use of $|D| \leq |B| - 1$ in conjunction with $\alpha > 4$. We conclude that strategy $D$ is an improvement in total. Hence, if a best response could be found in polynomial time, then a minimum size set cover could be found in polynomial time because our reduction is computable in polynomial time as well.
\end{proof}
\subsubsection*{\textbf{Approximate Equilibria}}
Since GE do not always exist and finding a best response is computationally hard, approximate equilibria are the only remaining option to construct almost stable networks for the practically important setting of Euclidean metrics in polynomial time.

To that end, we employ \(\Theta_k\)-graphs for 2D-Euclidean metrics. 
First introduced independently by Clarkson \cite{clarksonApproximation1987} and Keil \cite{keilApproximating1988}, \(\Theta_k\)-graphs are constructed as follows: Each node \(u\) partitions the plane into \(k\) disjoint cones with itself as the apex, each having an aperture of \(\frac{2\pi}{k}\). Then, for each cone, node~\(u\) adds an edge to the node~\(v\) whose projection onto the bisector of the cone is closest to~\(u\).

\(\Theta\)-routing is a way of selecting paths in a \(\Theta_k\)-graph. With this, to get from a node \(u\) to a node \(v\), the edge is used that \(u\) created into the cone that contains node~\(v\). This procedure is repeated for each hop until node~\(v\) is reached. First, we show that if \(k\) is too small, \(\Theta_k\)-graphs are not suited as approximate NEs. This result is well-known~\cite{narasimhanGeometric2007}, but we prove it for the sake of completeness.

\begin{restatable}{theorem}{thmfourthree}
	\label{the:euclid-theta-small-k}
    For every \(k \leq 5\), there exist 2D-Euclidean metrics where the \(\Theta_k\) graph with directed edges does not enable greedy routing.
\end{restatable}
\begin{proof}
	We construct a 2D-Euclidean metric by placing nodes in a polar coordinate system. For now, we assume that the reference direction of the coordinate system is equal to the bisector of some cone used in the construction of the \(\Theta_k\)-graph. The resulting metric is illustrated in \Cref{fig:theta-small-k}.
	
\begin{figure}[htb]
    \centering
    \resizebox{0.33\linewidth}{!}{
    \begin{tikzpicture}[bend angle=45,
        help/.style={circle, draw=gray, minimum size=20mm}],
    \node[vertex]   (u)										{\(u\)};
    
    \node[]   (b)			[above=6cm of u]				{}
        edge [dashed]	node[pos = 0, auto,swap, font=\LARGE] {cone bisector} (u);
    
    \node[vertex]   (v)			[above left=4cm and 2.80083cm of u]       		{\(v\)}
        edge []		node[auto,swap, font=\LARGE] {1}					(u);
    
    \node[]   (vv)			[above=4cm of u]				{}
        edge [dotted]	 (v);
    
    \node[vertex]   (w)			[above right=4.2cm and 2.94087cm of u]       		{\(w\)}
        edge []		node[auto, font=\LARGE] {\(1 + \epsilon\)}		(u);
    
    \node[]   (ww)			[above=4.2cm of u]				{}
        edge [dotted]	 (w);
        
	\draw pic[draw,angle radius=1.8cm,pic text=\(\frac{7\pi}{36}\), font=\LARGE] {angle=b--u--v};
	\draw pic[draw,angle radius=2cm,pic text=\(\frac{7\pi}{36}\), font=\LARGE] {angle=w--u--b};
	
	\draw pic[draw,dotted,angle radius=0.5cm,pic text=.] {angle=v--vv--u};
	\draw pic[draw,dotted,angle radius=0.5cm,pic text=.] {angle=u--ww--w};
    
    \end{tikzpicture}}
    \caption{2D Euclidean metric where there is no greedy path between \(u\) and \(w\) in the \(\Theta_k\)-graph for \(k \leq 5\).}
    \label{fig:theta-small-k}
\end{figure}
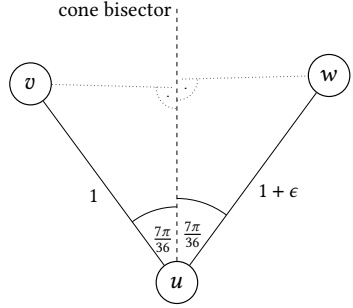
	
	Let \(u\) be located in the origin. Moreover, let \(v\) be at distance \(1\) and angle \(-\frac{7\pi}{36}\), and \(w\) at distance \(1 + \epsilon\) and angle \(\frac{7\pi}{36}\). We call this construction \(C\).
	Since both \(v\) and \(w\) are in the same cone for \(u\) if \(k \leq 5\) (because \(\frac{7\pi}{36} < \frac{\pi}{5}\)) and \(v\)'s projection on the bisector of that cone is closer to \(u\) than that of \(w\), agent \(u\) does not build an edge to \(w\). But since \(d(v,w) > d(u,w)\), agent \(u\) cannot use \(v\) on a greedy path to \(w\) and thus cannot have a greedy path to node~\(w\).
	
	In a slightly modified construction, even a global rotation of the cones (i.e.\ the reference direction of the coordinate system not lining up with the bisector of any cone) cannot achieve better results: First, we note that with cone-rotations of less than \(\frac{\pi}{k} - \frac{7\pi}{36}\) in either direction, the same cones stay occupied by nodes and thus the same argumentation still holds. Additionally, for any \(\beta \in \mathbb{R} \), a rotation by \(\beta\) is equal to a rotation by \(\beta \mod \frac{2\pi}{k}\). Thus, we place \(m = \frac{2\pi}{k} / (\frac{\pi}{k} - \frac{7\pi}{36})\) copies of \(C\), which we call \(C_1,... C_m\). Each \(C_{i+1}\) is rotated by \(\frac{\pi}{k} - \frac{7\pi}{36}\), compared to \(C_{i}\), and the distance between the centers of any \(C_i\) and \(C_{i'}\) is at least \(b = 4\). With this, in any rotation there is a \(u_i\) in some \(C_i\) for which their \(v_i\) and \(w_i\) are in the same cone. Since all nodes from other \(C_{i'}\) have a distance of at least \(b - 2\cdot d(u,w) = 2-2\epsilon\) to \(w_i\), agent \(u_i\) cannot use them on a greedy path \(w_i\) either and thus \(u\) has no greedy path to \(w_i\).
\end{proof}

With this limitation in mind, we give a general upper bound on the approximation ratio; although this leaves the case of \(k = 6\) open for now. Let \(f(k) = \frac{1}{1-2sin(\frac{\pi}{k})}\) be the maximum stretch of the \(\Theta\)-routing in a \(\Theta_k\)-graph with \(6 < k < n\) \cite{ruppertApproximating1991}.

\begin{theorem}
	\label{the:euclid-theta-upper}
    For \(6 < k < n\), every 2D-Euclidean instance of our game has a \(f(k) + \alpha\frac{k}{n-1}\)-NE.
\end{theorem}
\begin{proof}
Consider the \(\Theta_k\)-graph on \(\mathcal{P}\). With \(6 < k < n\), the stretch between any two nodes is at most \(f(k)\) with \(\Theta\)-routing, which gives a greedy path. The best response of an agent has costs of at least \(n-1\), because the stretch to every other node is at least \(1\). Thus, the \(\Theta_k\)-graph constitutes a \(f(k) + \alpha\frac{k}{n-1}\)-approximate NE.
\end{proof}
While this does not yield a constant approximation ratio, we note that, as \(n \to \infty\), the approximation ratio goes to \(f(k)\), which can be arbitrarily close to \(1\); for \(k \geq 15\), it is below \(2\).

Now, we establish our main result: the existence of a constant factor approximation that can be efficiently computed.

\begin{theorem}
	\label{the:euclid-theta8}
    Every 2D-Euclidean instance of the game has a \(5\)-approximate NE. 
\end{theorem}

\begin{proof}
Consider the \(\Theta_8\)-graph on \(\mathcal{P}\). By Bose,  De Carufel, Morin, van Renssen and Verdonschot~\cite{boseTight2016}, \(\Theta\)-routing in any \(\Theta_{4k+4}\)-graph gives stretches of at most \(1 + \frac{2 \sin(\frac{\pi}{4k+4})}{\cos(\frac{\pi}{4k+4}) - \sin(\frac{\pi}{4k+4})}\). Thus, the maximum stretch of \(\Theta\)-routing in the \(\Theta_8\)-graph is \(1 + \frac{2 \sin(\frac{\pi}{8})}{\cos(\frac{\pi}{8}) - \sin(\frac{\pi}{8})} = 1 + \sqrt{2}\), which yields a greedy path. In the best response of an agent, its stretch to any other agent is at least \(1\).

Every agent \(u\) for whom the other agents are not within a cone with angle at most \(\pi\) must build at least 2 edges in its best response. This stems from the fact that an edge to any node~\(v\) can only be part of a greedy path to a node~$i$ with \(\angle vui \leq \frac{\pi}{2}\), because otherwise
\[d(v,i) = \sqrt{d(u,i)^2 + d(u,v)^2 - 2d(u,i)d(u,v)\cos(\angle vui)} > \sqrt{d(u,i)^2} = d(u,i),\]
by the law of cosines and as such the path would not be a greedy path.

For every agent \(u\) for whom the other agents are within a cone with angle at most~\(\pi\), at least three cones of the \(\Theta_8\)-graph are empty and as such agent~\(u\) only builds at most five edges in the \(\Theta_8\)-graph.

Finally, we consider our stretchcost function. Edgecosts are at least \(\alpha\) or \(2\alpha\) in the best response, and at most \(5\alpha\) or \(8\alpha\) in the \(\Theta_8\)-graph, depending on whether all other agents are within a cone with angle at most \(\pi\) or not.

Let \(u \in \mathcal{P}\). Furthermore, let \(\s_\Theta\) be the strategy profiles of the \(\Theta_8\)-graph in a metric, where for agent~\(u\) not all other agents are within a cone with angle at most \(\pi\). Also, let \(\s^\pi_\Theta\) be the strategy profiles of the \(\Theta_8\)-graph where all other agents are within a cone with angle at most \(\pi\). Let \(\s_{br}\) and \(\s^\pi_{br}\) be the corresponding strategy profiles, where agent~\(u\)'s strategy is changed to its best response, while all other strategies stay unchanged.

For our stretchcost function, we get that the \(\Theta_8\)-graph is a
\[\max\left(\frac{c_u(s_\Theta)}{c_u(s_{br})}, \frac{c_u(s^\pi_\Theta)}{c_u(s^\pi_{br})}\right) =
\max\left( \frac{8\alpha + (1 + \sqrt{2})(n-1)}{2\alpha + (n-1)}, \frac{5\alpha + (1 + \sqrt{2})(n-1)}{\alpha + (n-1)}\right) \leq 5\]
-approximate NE, because the stretch to all \(n-1\) other agents is at least \(1\) in a best response and at most \(1 + \sqrt{2}\) in the \(\Theta_8\)-graph.
\end{proof}

Note that \Cref{the:euclid-theta8} does not imply a constant bound for 1-2 metrics, since there are instances in 1-2 metrics that cannot be embedded into the Euclidean plane.

The following theorem gives a lower bound on the approximation ratio of \(\Theta_k\)-graphs. For the \(\Theta_8\)-graphs we used in the last theorem, this bound is tight. 
\begin{theorem}
	\label{the:euclid-theta-lower}
    There are 2D-Euclidean instances of our game, where the \(\Theta_k\)-graph is not a \((\ceil{\frac{k}{2}} + 1 - \epsilon)\)-approximate NE.
\end{theorem}
\begin{proof}
	We construct a 2D-Euclidean metric by placing nodes in a polar coordinate system. For now, we assume that the reference direction of the coordinate system is equal to either the bisector of some cone used in the construction of the \(\Theta_k\)-graph if \(\ceil{\frac{k}{2}} + 1\) is odd or the edge of a cone otherwise.
	Let \(u\) be located at the origin and let \(v\) be at distance \(1\) and angle \(0\). Let the nodes \(w_1, ..., w_{\ceil{\frac{k}{2}} + 1}\) be at angles evenly spaced out over the interval \([-\frac{\pi}{2} + \frac{2\pi}{5k}; \frac{\pi}{2} - \frac{2\pi}{5k}]\) and at distances \(\frac{1}{\cos(\angle vuw_i)}\) (this is possible, as \(\angle vuw_i < \frac{\pi}{2}\) and thus \(\cos(\angle vuw_i) > 0\)). We call this construction \(C\).
	
	Agent \(u\) builds edges to \(\ceil{\frac{k}{2}} + 1\) nodes in the \(\Theta_k\)-graph because each \(w_i\) is in a different cone. For all \(i\), we have, by the law of cosines, that
	\[d(v, w_i)
	= \sqrt{1^2 + d(u, w_i)^2 - 2 \cdot 1 \cdot d(u, w_i) \cos(\angle vuw_i)}
	= \sqrt{d(u, w_i)^2 - 1}
	< d(u,w_i).\]
	As such, agent \(v\) cannot use \(u\) on greedy paths to any \(w_i\), whereas \(u\) can use \(v\) for all of them. If \(v\) would not have a greedy path to some \(w_i\), this could not be a \((\ceil{\frac{k}{2}} + 1 - \epsilon)\)-approximate NE because adding edges to all \(w_i\) without a greedy path would improve \(v\)'s costs from at least \(Z\) to less than \(Z\), which by definition of \(Z\) is a large improvement, also by more than a factor of \((\ceil{\frac{k}{2}} + 1 - \epsilon)\). Thus, agent~\(u\) could replace all of its edges with an edge to \(v\) and still retain greedy paths to all nodes. Let \(sc\) and \(sc'\) be agent~\(u\)'s stretchcosts before and after this move. Agent~\(u\)'s total costs would improve by a factor of \(\frac{(\ceil{\frac{k}{2}} + 1)\alpha + sc'}{\alpha + sc}\) which is at least \(\ceil{\frac{k}{2}} + 1 - \epsilon\) for \(\alpha \geq \frac{(\ceil{\frac{k}{2}} + 1)sc - sc' - \epsilon}{\epsilon}\).
	
	In a slightly modified construction, even a global rotation of the cones (i.e.\ the reference direction of the coordinate system not lining up with the bisector/edge of any cone) cannot achieve better results: First, we note that with cone-rotations of less than \(\frac{4\pi}{5k}\) in either direction, the same cones stay occupied by nodes and as such, this construction still holds. Also, a rotation by \(\beta\) is equal to a rotation by \(\beta \mod \frac{2\pi}{k}\). Thus, we place two copies of \(C\), which we call \(C_1\) and \(C_2\). The copy \(C_2\) is rotated by \(i\frac{2\pi}{k} + \frac{\pi}{k}\), for some \(i \in \mathbb{N}\), compared to \(C_1\), and its center is displaced. With this, in any rotation, one of \(u_1\) or \(u_2\) still has to build  \(\ceil{\frac{k}{2}} + 1\) many edges. Additionally, by placing the \(C_1\) and \(C_2\) far enough apart and choosing \(i\) for the rotation to ensure that they are each contained in cones of \(u\) of the other that are occupied by some \(w_i\), these do not impact the number of edges agent~\(u\) builds in the \(\Theta_k\)-graph or needs to build in its best response.
\end{proof}
Thus, we cannot get a better bound on the approximation ratio of \(\Theta_8\)-graphs.
However, the approximation factor does not improve by choosing a smaller \(k\): With \(k = 7\) being an odd number, \Cref{the:euclid-theta-lower} still gives the same lower bound of \(5 - \epsilon\). With \(k = 6\) the known bound on the stretch goes up to at least \(7\) (and that is not along necessarily greedy paths; to get greedy paths the stretch might be as large as the bound on \(\Theta\)-routing of \(12\sqrt{3}\))~\cite{akitayaSpanning2022}. For smaller values of \(k\), by \Cref{the:euclid-theta-small-k}, there might be no greedy paths between some nodes.

\section{General Metrics} \label{sec:gemetrics}
Now we consider general metric spaces. Naturally, all negative results from the preceding sections, like hardness and non-convergence results, immediately carry over. Specifically, the results that are not implicit in the proofs in this section is that computing best responses is NP-hard, that GE are not \(\Omega(\frac{\alpha n}{\alpha + n})\) -approximate NE and that GE may not exist.

\subsubsection*{\textbf{Hardness of the Existence of Equilibria}} 
Here, we show that deciding whether GE and NE exist in a given instance with an arbitrary metric is NP-hard.

\begin{restatable}{theorem}{thmfiveone}
    \label{the:existenceNP}
    In general metric spaces, is both NP-hard to decide, whether an instance admits a NE and whether it admits a GE.
\end{restatable}
\begin{proof}
We modify the proof for Theorem 6.1 from Moscibroda, Schmid and Wattenhofer~\cite{moscibrodaTopologies2006} for \(k = 1\) by changing the distances but retaining the same general structure of the space.

We reduce from the NP-complete \textsc{3-SAT} variant, where each variable occurs in at most three clauses (with clauses with fewer than three literals being allowed)~\cite{toveySimplified1984}.

Let \(\mathcal{I}\) be such a \textsc{3-SAT} instance with variables \(\mathcal{X}\) and clauses \(\mathcal{C}\).
We use the same reduction function for both NE and GE and first show that the resulting instance has no GE, if there is no satisfying assignment of \(\mathcal{I}\) and then, that it has a NE otherwise. This suffices, as every NE is a GE.

We construct a metric space with nodes \(\{y, z, d\} \cup \{a_i, b_i, c_i\ |\ i \leq |\mathcal{C}|\} \cup \{t_u, f_u\ |\ u \leq |\mathcal{X}|\}\) with distances as shown in \Cref{fig:existenceNP}, such that \(d(c_i, t_u) = 1.6 - \epsilon\) for all positive literals \(x_u\) in \(C_i\) and \(d(c_i, f_u) = 1.6 - \epsilon\) for all negative literals \(\overline{x_u}\) in \(C_i\). Furthermore, let \(\alpha = 0.6\).

\begin{figure}[htb]
    \centering
    \resizebox{0.58\linewidth}{!}{
    \begin{tikzpicture}[trim left=-0.5cm,bend angle=60],
    \node[vertex]   (a)                                           {\(a_1\)};
    \node[vertex]   (b)        [right=3cm of a.center]                 {\(b_1\)}
        edge []                          node[auto,swap] {1.14} (a);
    \node[vertex]   (c)        [right=3cm of b.center]                 {\(c_1\)}
        edge []                          node[auto,swap] {1} (b);
        
    \node[vertex]   (a2)        [above=1cm of a.center]                             {\(a_2\)};
    \node[vertex]   (b2)        [right=3cm of a2.center]                 {\(b_2\)}
        edge []                          node[auto,swap] {1.14} (a2);
    \node[vertex]   (c2)        [right=3cm of b2.center]                 {\(c_2\)}
        edge []                          node[auto,swap] {1} (b2);
        
    \node[vertex]   (a3)        [above=1cm of a2.center]                             {\(a_3\)};
    \node[vertex]   (b3)        [right=3cm of a3.center]                 {\(b_3\)}
        edge []                          node[auto,swap] {1.14} (a3);
    \node[vertex]   (c3)        [right=3cm of b3.center]                 {\(c_3\)}
        edge []                          node[auto,swap] {1} (b3);
        
    \node[vertex]   (1)        [below left=3cm and 1.5cm of b.center]                 {\(y\)}
        edge [dotted] (b2)
        edge [dotted,bend left=10] (b3)
        edge [dotted] (c2)
        edge [dotted] (c3)
        edge []                          node[auto] {1.96} (a)
        edge [dotted,bend left] (a2)
        edge [dotted,bend left] (a3)
        edge []                          node[pos=0.2,fill=white] {\(2-\epsilon\)} (b)
        edge []                          node[pos=0.8,fill=white] {\(2+2\epsilon\)} (c);
    \node[vertex]   (2)        [right=3cm of 1.center]                 {\(z\)}
        edge [dotted] (a2)
        edge [dotted] (a3)
        edge [dotted] (b2)
        edge [dotted,bend right=10] (b3)
        edge []                          node[pos=0.8,swap,fill=white] {\(2+2\epsilon\)} (a)
        edge []                          node[pos=0.2,fill=white,swap] {2} (b)
        edge []                          node[auto,swap] {\(2+\epsilon\)} (c)
        edge [dotted,bend right] (c2)
        edge [dotted,bend right] (c3)
        edge []                          node[auto,swap] {\(1-2\epsilon\)} (1);
    
    \node[vertex]   (t1)        [above right=1 cm and 4cm of 2.center]                 {\(t_1\)}
        edge [ bend left=20]            node[auto] {\(1.6\)} (2)
        edge [ bend left=20]            node[auto,swap] {\(1.6 - \epsilon\)} (c);
    \node[vertex]   (f1)        [right=1.2cm of t1.center]                 {\(f_1\)}
        edge [dotted] (c2)
        edge [dotted] (c3)
        edge [ bend left=30]            node[auto] {\(1.6\)} (2)
        edge []                          node[auto] {\(0.85\)} (t1);
    
    \node[vertex]   (t2)        [right=1.2cm of f1.center]                 {\(t_2\)}
        edge [dotted] (c2)
        edge [ bend left=40]            node[auto] {\(1.6\)} (2);
    \node[vertex]   (f2)        [right=1.2cm of t2.center]                 {\(f_2\)}
        edge [ bend left=50]            node[auto] {\(1.6\)} (2)
        edge []                          node[auto] {\(0.85\)} (t2)
        edge [ bend right=10]           node[below,pos=0.8] {\(1.6 - \epsilon\)} (c);
    
    \node[vertex]   (t3)        [right=1.2cm of f2.center]                 {\(t_3\)}
        edge [dotted] (c3)
        edge [ bend left=60]            node[auto] {\(1.6\)} (2);
    \node[vertex]   (f3)        [right=1.2cm of t3.center]                 {\(f_3\)}
        edge [ bend left=70]            node[auto] {\(1.6\)} (2)
        edge []                          node[auto] {\(0.85\)} (t3)
        edge [ bend right=20]           node[above,pos=0.85] {\(1.6 - \epsilon\)} (c);
        
    \node[vertex]   (d)        [above right=5cm of c3.center]                 {\(d\)}
        edge [dotted] (a)
        edge [dotted,bend right=5] (b)
        edge [dotted] (c)
        edge [dotted,bend right=5] (a2)
        edge [dotted] (b2)
        edge [dotted] (c2)
        edge []                          node[auto] {\(1.3\)} (a3)
        edge []                          node[auto] {\(1.3\)} (b3)
        edge []                          node[auto] {\(1.3\)} (c3)
        edge []                          node[pos=0.6,right] {\(1.3\)} (t1)
        edge []                          node[pos=0.6,right] {\(1.3\)} (f1)
        edge []                          node[pos=0.6,right] {\(1.3\)} (t2)
        edge []                          node[pos=0.6,right] {\(1.3\)} (f2)
        edge []                          node[pos=0.6,right] {\(1.3\)} (t3)
        edge []                          node[pos=0.6,right] {\(1.3\)} (f3);

    \end{tikzpicture}}
    \caption{Metric for 3-SAT formula
    \((x_1 \vee \overline{x_2} \vee \overline{x_3}) \wedge
    (\overline{x_1} \vee x_2) \wedge
    (\overline{x_1} \vee x_3)\).
    Only the distances between \(c_i\) and \(t_u\)/\(f_u\) depend on the formula; all other distances given are the same for all \(i, u\) for \(a_i, b_i, c_i, t_u\) and \(f_u\). Dotted lines represent distances analogous to those between similar nodes. All other distances not given in this figure are as the upper bounds given by the space being metric.}
    \label{fig:existenceNP}
\end{figure}
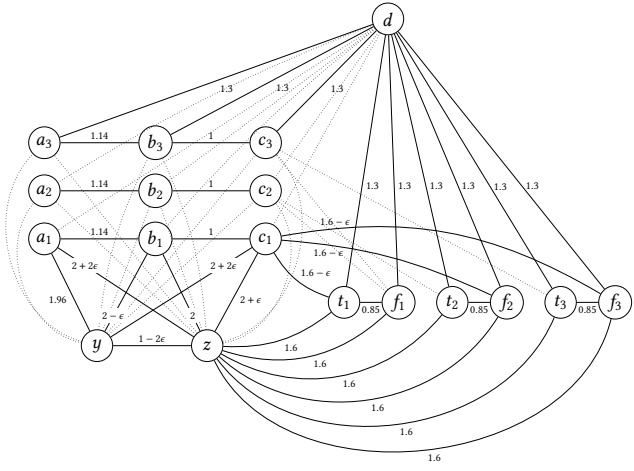

First, we note some general properties of NE and GE in this instance.
We show that for every \(u \leq |\mathcal{X}|\), agent \(z\) builds exactly one edge to either \(t_u\) or \(f_u\) in any GE.
Assume at first, that \(z\) does not build an edge to either of these nodes. In this case, agent \(z\)'s stretch to both nodes is at least \(\frac{d(z,c_i) + d(c_i,t_u)}{d(z,t_u)} = \frac{3.6}{1.6} > 1 + \alpha\) (each via some \(c_i\) that is at distance \(1.6 - \epsilon\) to it) and adding one of the edges would be an improving move.
Now, assume that \(z\) builds an edge to both nodes. Since every variable appears in at most three clauses, one of \(t_u\) or \(f_u\) has at most one \(c_i\) at distance \(1.6 - \epsilon\). If one of them has no \(c_i\) at distance \(1.6 - \epsilon\), removing the edge to it would only increase the stretch to it by \(\frac{d(z,f_u) + d(f_u,t_u)}{d(z,t_u)} - 1 = \frac{2.45}{1.6} - 1 < \alpha\). Let, w.l.o.g., \(t_u\) only have \(c_i\) at distance \(1.6 - \epsilon\). If \(z\) has a path of length at most \(3.2 - \epsilon\) to \(c_i\) that does not include \(t_u\), then agent~\(z\) could remove the edge to \(t_u\) and only the stretch to \(t_u\) would increase from \(1\) to \(\frac{d(z,f_u) + d(f_u,t_u)}{d(z,t_u)} = \frac{2.45}{1.6}\) which is an increase of less than \(\alpha\) while the stretch to \(c_i\) would not increase. If \(z\) does not have a path of length at most \(3.2 - \epsilon\) to \(c_i\), that does not include \(t_u\), then agent~\(z\) could swap the edge from \(t_u\) to \(c_i\) to increase the stretch to \(t_u\) from \(1\) to \(\frac{d(z,f_u) + d(f_u,t_u)}{d(z,t_u)} = \frac{2.45}{1.6}\) and decrease the stretch to \(c_i\) from at least \(\frac{3.2 - \epsilon}{d(z,c_i)} = \frac{3.2 - \epsilon}{2+\epsilon}\) to \(1\) which is an improvement in total.

Thus, we have that for every \(u \leq |\mathcal{X}|\), agent \(z\) builds exactly one edge to either \(t_u\) or \(f_u\) in any GE.

We also note that in any GE, for all \(i, u\), the edges \((b_i, c_i)\), \((c_i, b_i)\), \((b_i, a_i)\), \((y, z)\), \((z,y)\), \((t_u, f_u)\), \((f_u, t_u)\), \((a_i, d)\), \((b_i, d)\), \((c_i, d)\), \((t_u, d)\) and \((f_u, d)\) have to be built to enable greedy routing because there are, respectively, no other nodes closer to the second node than the first node itself.
Furthermore, every \(a_i\) also builds an edge to \(b_i\): If there was no edge to \(b_i\), agent \(a_i\) would need to build an edge to \(c_i\) to have a greedy path to \(b_i\). Then, swapping that edge to \(b_i\) would decrease the stretch to \(b_i\) by \(\frac{d(a_i,c_i) + d(c_i,b_i)}{d(a_i,b_i)} - 1 = \frac{3.14}{1.14} - 1 > 0\) while not changing the stretch to \(c_i\) and thus be an improving move.

Agent \(y\) also builds an edge to every \(a_i\) because the only other possible greedy paths via \(b_i\) and \(d\) have stretches of \(\frac{d(y,b_i) + d(b_i, a_i)}{d(y,a_i)} = \frac{3.14 - \epsilon}{1.96} > 1 + \alpha\) and \(\frac{d(y,d) + d(d, a_i)}{d(y,a_i)} = \frac{4.56}{1.96} > 1 +\alpha\). Agent \(y\) does not build an edge to both \(b_i\) and \(c_i\), for any \(i\), because if it would do so, it could simply remove the edge to \(c_i\) and increase stretches by at most \(\frac{d(y,b_i) + d(b_i, c_i)}{d(y,c_i)} - 1 = \frac{3-\epsilon}{2+2\epsilon} - 1 < \alpha\), because the shortest paths to any \(t_u\) and \(f_u\) use \(z\) in either case. In fact, agent \(y\) does not build an edge to any \(c_i\), because if it would do so, it could swap its edge from \(c_i\) to \(b_i\), changing the sum of stretches by at most \[\frac{d(y,b_i) + d(b_i, x_i)}{d(y,x_i)} - \frac{d(y,c_i) + d(c_i, b_i)}{d(y,b_i)} = \frac{3-\epsilon}{2+2\epsilon} - \frac{3+2\epsilon}{2-\epsilon} < 0.\]

Now, let there be no satisfying assignment for \(\mathcal{I}\). We show that there is no GE in our instance. Assume for the sake of contradiction that there is a GE.
Since for every \(u \leq |\mathcal{X}|\), agent \(z\) builds exactly one edge to either \(t_u\) or \(f_u\) in any GE and there is no satisfying assignment, there has to be some \(c_i\), where \(z\) does not have an edge to a \(t_u\) or \(f_u\) at distance \(1.6 - \epsilon\). Let \(i\) be such that this is true for \(c_i\). 
For there to be a greedy path to \(c_i\), agent \(z\) needs to build an edge to \(b_i\), \(c_i\) or \(d\). In fact, if \(z\) does not build an edge to \(b_i\) or \(c_i\) but only to \(d\), the stretch to \(c_i\) would be \(\frac{d(z,d) + d(d,c_i)}{d(z,c_i)} = \frac{4.1}{2+\epsilon} > 1 + \alpha\). As such, agent \(z\) builds an edge to \(b_i\) or \(c_i\). Agent \(z\) does not build both edges because if it did, it could remove its edge to \(c_i\) and increase stretches by \[\frac{d(z,b_i) + d(b_i,c_i)}{d(z,c_i)} - 1 = \frac{3}{2+\epsilon} - 1 < \alpha.\] Agent \(z\) also does not build an edge to \(a_i\), because the stretch to \(a_i\) via \(y\) is already \(\frac{d(z,y) + d(y,a_i)}{d(z,a_i)} = \frac{2.96-2\epsilon}{2+2\epsilon} < 1 + \alpha\) and the edge to \(a_i\) would not allow for any shorter greedy paths to \(b_i\) or \(c_i\), because \(z\) needs an edge to \(b_i\) or \(c_i\) either way.

When looking at the strategies of \(y\) and \(z\) pertaining to \(a_i, b_i\) and \(c_i\), we again only get the four cases illustrated in \Cref{fig:candidate}. 
We examine all of these possible cases and show that they cannot be part of a GE.

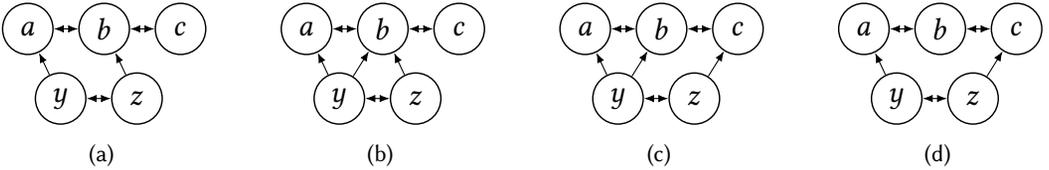
\begin{figure}[h]
    \centering
    \begin{subfigure}[b]{0.2\textwidth}
         \centering
         
        \resizebox{\linewidth}{!}{
        \begin{tikzpicture}[bend angle=90],
        \node[vertex]   (a)                                     {\(a\)};
        \node[vertex]   (b)        [right=1cm of a.center]                 {\(b\)}
            edge [<->]                          (a);
        \node[vertex]   (c)        [right=1cm of b.center]                 {\(c\)}
            edge [<->]                          (b);
            
        \node[vertex]   (1)        [below left=1cm and 0.5cm of b.center]                 {\(y\)}
            edge [->]                          (a);
        \node[vertex]   (2)        [right=1cm of 1.center]                 {\(z\)}
            edge [->]                          (b)
            edge [<->]                          (1);
        \end{tikzpicture}}
        \caption{}
        \label{fig:candidatea}
    \end{subfigure}
    \hfill
    \begin{subfigure}[b]{0.2\textwidth}
         \centering
        \resizebox{\linewidth}{!}{
        \begin{tikzpicture}[bend angle=90],
        \node[vertex]   (a)                                     {\(a\)};
        \node[vertex]   (b)        [right=1cm of a.center]                 {\(b\)}
            edge [<->]                          (a);
        \node[vertex]   (c)        [right=1cm of b.center]                 {\(c\)}
            edge [<->]                          (b);
            
        \node[vertex]   (1)        [below left=1cm and 0.5cm of b.center]                 {\(y\)}
            edge [->]                          (a)
            edge [->]                          (b);
        \node[vertex]   (2)        [right=1cm of 1.center]                 {\(z\)}
            edge [->]                          (b)
            edge [<->]                          (1);
        \end{tikzpicture}}
        \caption{}
        \label{fig:candidatec}
    \end{subfigure}
    \hfill
    \begin{subfigure}[b]{0.2\textwidth}
         \centering
        \resizebox{\linewidth}{!}{
        \begin{tikzpicture}[bend angle=90],
        \node[vertex]   (a)                                     {\(a\)};
        \node[vertex]   (b)        [right=1cm of a.center]                 {\(b\)}
            edge [<->]                          (a);
        \node[vertex]   (c)        [right=1cm of b.center]                 {\(c\)}
            edge [<->]                          (b);
            
        \node[vertex]   (1)        [below left=1cm and 0.5cm of b.center]                 {\(y\)}
            edge [->]                          (a)
            edge [->]                          (b);
        \node[vertex]   (2)        [right=1cm of 1.center]                 {\(z\)}
            edge [->]                          (c)
            edge [<->]                          (1);
        \end{tikzpicture}}
        \caption{}
        \label{fig:candidated}
    \end{subfigure}
    \hfill
    \begin{subfigure}[b]{0.2\textwidth}
         \centering
        \resizebox{\linewidth}{!}{
        \begin{tikzpicture}[bend angle=90],
        \node[vertex]   (a)                                     {\(a\)};
        \node[vertex]   (b)        [right=1cm of a.center]                 {\(b\)}
            edge [<->]                          (a);
        \node[vertex]   (c)        [right=1cm of b.center]                 {\(c\)}
            edge [<->]                          (b);
            
        \node[vertex]   (1)        [below left=1cm and 0.5cm of b.center]                 {\(y\)}
            edge [->]                          (a);
        \node[vertex]   (2)        [right=1cm of 1.center]                 {\(z\)}
            edge [->]                          (c)
            edge [<->]                          (1);
        \end{tikzpicture}}
        \caption{}
        \label{fig:candidateb}
    \end{subfigure}
    \caption{Candidates for equilibrium strategies for vertices \(y\) and \(z\)}
    \label{fig:candidate}
\end{figure}

\textbf{Case 1:} In \Cref{fig:candidatea}, if \(y\) adds an edge to \(b\), the stretch to \(b\) decreases from \(\frac{d(y,a) + d(a,b)}{d(y,b)} = \frac{3.1}{2 - \epsilon}\) to \(1\) and the stretch to \(c\) from \(\frac{d(y,z) + d(z,b) + d(b,c)}{d(y,c)} = \frac{4-2\epsilon}{2+2\epsilon}\) to \(\frac{d(y,b) + d(b,c)}{d(y,c)} = \frac{3-\epsilon}{2+2\epsilon}\), which is a total improvement of more than \(\alpha\).

\textbf{Case 2:} In \Cref{fig:candidatec}, agent \(z\) can swap its edge from \(b\) to \(c\), changing the sum of stretches by \(\frac{d(z,y) + d(y,b)}{d(z,b)} - \frac{d(z,b) + d(b,c)}{d(z,c)} = \frac{3-3\epsilon}{2} - \frac{3}{2+\epsilon} < 0\).

\textbf{Case 3:} In \Cref{fig:candidated}, if \(y\) removes its edge to \(b\), the stretch to \(c\) stays unchanged at \(\frac{d(y,z) + d(z,c)}{d(y,c)} = \frac{d(y,b) + d(b,c)}{d(y,c)}\), while the stretch to \(b\) increases from \(1\) to \(\frac{d(y,a) + d(a,b)}{d(y,b)} = \frac{3.1}{2  - \epsilon}\) which is an increase of less than \(\alpha\).

\textbf{Case 4:} In \Cref{fig:candidateb}, agent \(z\) can swap its edge from \(c\) to \(b\), changing the sum of stretches by \(\frac{d(z,b) + d(b,c)}{d(z,c)} - \frac{d(z,c) + d(c,b)}{d(z,b)} = \frac{3}{2 + \epsilon} - \frac{3+\epsilon}{2} < 0\).

Thus, there can be no GE in this instance if there is no satisfying assignment for \(\mathcal{I}\).

Now, let \(A_{\mathcal{I}}\) be a assignment satisfying \(\mathcal{I}\). We construct a NE, as illustrated in \Cref{fig:existenceNPEquilibrium}.
\begin{itemize}
\item Agent \(y\) builds edges to \(z\) and every \(a_i\) and \(b_i\).

\item Agent \(z\) builds edges to \(y\) and every \(t_u\) where \(x_u\) is true in \(A_{\mathcal{I}}\) and every \(f_u\) where \(x_u\) is false in \(A_{\mathcal{I}}\).

\item Agent \(d\) builds edges to every  \(a_i, b_i, c_i, t_u\) and \(f_u\).

\item Every \(a_i\) build edges to \(d, y\) and its corresponding \(b_i\).

\item Every \(b_i\) build edges to \(d\) and its corresponding \(a_i\) and \(c_i\).

\item Every \(c_i\) build edges to \(d, z\), its corresponding \(b_i\) and every \(t_u\) and \(f_u\) at distance \(1.6 - \epsilon\).

\item Every \(t_u\) and \(f_u\) builds edges to \(d, z\), its corresponding \(f_u\) or \(t_u\) and every \(c_i\) at distance \(1.6 - \epsilon\).
\end{itemize}

\begin{figure}[htb]
    \centering
    
    \resizebox{0.58\linewidth}{!}{
    \begin{tikzpicture}[trim left=-0.5cm,bend angle=60],
    \node[vertex]   (a)                                           {\(a_1\)};
    \node[vertex]   (b)        [right=3cm of a.center]                 {\(b_1\)}
        edge [<->] (a);
    \node[vertex]   (c)        [right=3cm of b.center]                 {\(c_1\)}
        edge [<->] (b);
        
    \node[vertex]   (a2)        [above=1cm of a.center]                             {\(a_2\)};
    \node[vertex]   (b2)        [right=3cm of a2.center]                 {\(b_2\)}
        edge [<->] (a2);
    \node[vertex]   (c2)        [right=3cm of b2.center]                 {\(c_2\)}
        edge [<->] (b2);
        
    \node[vertex]   (a3)        [above=1cm of a2.center]                             {\(a_3\)};
    \node[vertex]   (b3)        [right=3cm of a3.center]                 {\(b_3\)}
        edge [<->] (a3);
    \node[vertex]   (c3)        [right=3cm of b3.center]                 {\(c_3\)}
        edge [<->] (b3);
        
    \node[vertex]   (1)        [below left=3cm and 1.5cm of b.center]                 {\(y\)}
        edge [<->] (a)
        edge [<->,dotted,bend left] (a2)
        edge [<->,dotted,bend left] (a3)
        edge [->] (b)
        edge [->,dotted,bend left=5] (b2)
        edge [->,dotted,bend left=10] (b3);
    \node[vertex]   (2)        [right=3cm of 1.center]                 {\(z\)}
        edge [<->]                  (1)
        edge [<-]                   (c)
        edge [<-,dotted,bend left=5] (c2)
        edge [<-,dotted,bend left=10] (c3);
    
    \node[vertex]   (t1)        [above right=1 cm and 4cm of 2.center]                 {\(t_1\)}
        edge [<->, bend left=20] (c)
        edge [->, bend left=20] (2);
    \node[vertex]   (f1)        [right=1cm of t1.center]                 {\(f_1\)}
        edge [<->,dotted] (c2)
        edge [<->,dotted] (c3)
        edge [<->, bend left=30] (2)
        edge [<->]               (t1);
    
    \node[vertex]   (t2)        [right=1.5cm of f1.center]                 {\(t_2\)}
        edge [<->,dotted] (c2)
        edge [<->, bend left=40]  (2);
    \node[vertex]   (f2)        [right=1cm of t2.center]                 {\(f_2\)}
        edge [<->]                (t2)
        edge [->, bend left=50]   (2)
        edge [<->, bend right=10] (c);
    
    \node[vertex]   (t3)        [right=1.5cm of f2.center]                 {\(t_3\)}
        edge [<->,dotted] (c3)
        edge [->, bend left=60]  (2);
    \node[vertex]   (f3)        [right=1cm of t3.center]                 {\(f_3\)}
        edge [<->, bend left=70]  (2)
        edge [<->]                (t3)
        edge [<->, bend right=20] (c);
        
    \node[vertex]   (d)        [above right=5cm of c3.center]                 {\(d\)}
        edge [<->,dotted] (a)
        edge [<->,dotted,bend right=5] (b)
        edge [<->,dotted] (c)
        edge [<->,dotted,bend right=5] (a2)
        edge [<->,dotted] (b2)
        edge [<->,dotted] (c2)
        edge [<->]                          (a3)
        edge [<->]                          (b3)
        edge [<->]                          (c3)
        edge [<->]                          (t1)
        edge [<->]                          (f1)
        edge [<->]                          (t2)
        edge [<->]                          (f2)
        edge [<->]                          (t3)
        edge [<->]                          (f3);

    \end{tikzpicture}}
    \caption{NE for 3-SAT formula
    \((x_1 \vee \overline{x_2} \vee \overline{x_3}) \wedge
    (\overline{x_1} \vee x_2) \wedge
    (\overline{x_1} \vee x_3)\), assigning \(x_2\) as true and \(x_1\) and \(x_3\) as false. Dotted lines also represent edges.}
    \label{fig:existenceNPEquilibrium}
\end{figure}
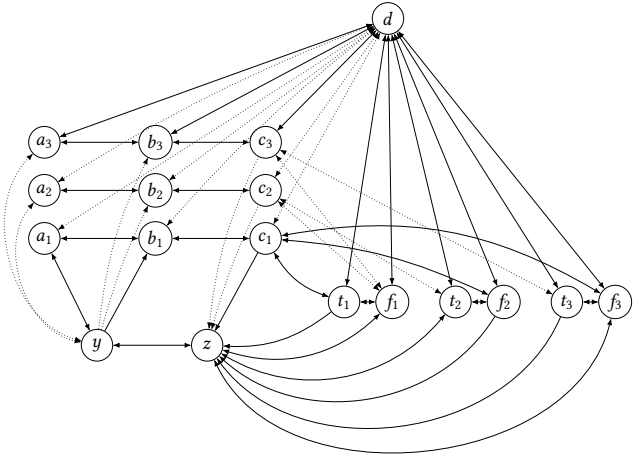

We show that no agent has an incentive to change their strategy:

\textbf{Agent \(y\):} As noted above, the edges to \(z\) and any \(a_i\) cannot be removed and the edges to any \(c_i\) cannot be added. Thus, the only possible improving changes are dropping the edge to some \(b_i\), adding an edge to \(d\) or adding an edge to some \(t_u\) or \(f_u\) that \(z\) does not have an edge to. Dropping an edge to a \(b_i\) would increase stretches by \(\frac{d(y,a) + d(a,b_i)}{d(y,b_i)} - 1 + \frac{d(y,a) + d(a,b_i) + d(b_i,c_i)}{d(y,c_i)} - \frac{d(y,b_i) + d(b_i,c_i)}{d(y,c_i)} = \frac{3.1}{2-\epsilon} - 1 + \frac{4.1}{2 + 2\epsilon} - \frac{3-\epsilon}{2 + 2\epsilon} < \alpha\). Adding an edge to a \(t_u\) or \(f_u\) that \(z\) does not have an edge to would decrease stretches by \(\frac{d(y,z) + d(z,f_u) + d(f_u,t_u)}{d(y,t_u)} - 1 = \frac{3.45 - 2\epsilon}{2.6 - 2\epsilon} - 1 < \alpha\). Since both of these effects are independent, doing both would not be an improving move either. Adding an edge to \(d\) would not decrease the stretch to any node even in combination with another change.

\textbf{Agent \(z\):} As noted above, the edge to \(y\) cannot be removed. For every \(c_i\), \(z\) has an edge to a \(t_u\) or \(f_u\) at distance \(1.6 - \epsilon\) because \(A_{\mathcal{I}}\) is a satisfying assignment. Since \(z\) has to have an edge to exactly one of \(t_u\) and \(f_u\) for each \(u\), adding or removing edges to any \(t_u\) or \(f_u\) cannot be an improving move either. Swapping an edge between a corresponding \(t_u\) and \(f_u\) only swaps their stretches and can not decrease any other stretches. The stretch to \(d\) is already \(1\) and an edge to it cannot be part of any other shorter greedy paths either. Thus, the only possible improving changes are adding an edge to some \(a_i, b_i\) or \(c_i\). Since all of these moves just add edges, if none of these is an improving move individually, a combination of them cannot be an improving move either. Adding an edge to some \(a_i\) would decrease stretches by \(\frac{d(z,y) + d(y,a_i)}{d(z,a_i)} - 1 = \frac{2.96-2\epsilon}{2 + 2\epsilon} - 1 < \alpha\). Adding an edge to some \(b_i\) would decrease stretches by \(\frac{d(z,y) + d(y,b)}{d(z,b)} - 1 + \frac{d(z, t_u) + d(t_u, c_i)}{d(z, c_i)} - \frac{d(z, b_i) + d(b_i, c_i)}{d(z, c_i)} = \frac{3-3\epsilon}{2} - 1 + \frac{3.2 - \epsilon}{2 + \epsilon} - \frac{3}{2 + \epsilon} < \alpha\). Adding an edge to some \(c_i\) would decrease stretches by \(\frac{d(z, t_u) + d(t_u, c_i)}{d(z, c_i)} - 1 = \frac{3.2 - \epsilon}{2 + \epsilon} - 1 < \alpha\).

\textbf{Agent \(d\):} Dropping any edge would increase the stretch to that node by more than \(\alpha\) because the distance between any two nodes and thus the increase in distance if \(d\) would drop an edge is more than \(1.3 \alpha = 0.78\). The stretch to \(y\) and \(z\) is already \(1\) via another node and an edge to them cannot be part of any other shorter greedy paths either. Thus, building any additional edges cannot improve costs either.

\textbf{Agent \(a_i\):} As noted above, the edges to \(d\) and any \(b_i\) cannot be removed and the stretch to any node other than \(z\) is already \(1\) and an edge to them cannot be part of any other shorter greedy paths either. Thus, the only possible improving changes are removing the edge to \(y\) and adding one to \(z\). Doing both would change stretches by \(\frac{d(a_i,y) + d(y,z)}{d(a_i,z)} - \frac{d(a_i,z) + d(z,y)}{d(a_i,y)} = \frac{2.96 - 2\epsilon}{2 + 2\epsilon} - \frac{3}{1.96} < 0\). Dropping the edge to \(y\) would remove the only greedy path to \(y\). Adding the edge to \(z\) would decrease only the stretch to it by \(\frac{d(a_i,y) + d(y,z)}{d(a_i,z)} - 1 = \frac{2.96 - 2\epsilon}{2 + 2\epsilon} - 1 < \alpha\).

\textbf{Agent \(b_i\):} As noted above, the edges to \(d\) and any \(a_i\) and \(c_i\) cannot be removed and the stretch to any node other than \(y\) and \(z\) is already \(1\) and an edge to them cannot be part of any other shorter greedy paths either. Thus, the only possible improving changes are adding edges to \(y\) or \(z\). Adding the edge to \(y\) would decrease stretches by \(\frac{d(b_i,a_i) + d(a_i,y)}{d(b_i,y)} - 1 + \frac{d(b_i,c_i) + d(c_i,z)}{d(b_i,z)} - \frac{d(b_i,y) + d(y,z)}{d(b_i,z)} = \frac{3.1}{2 - \epsilon} - 1 + \frac{3 + \epsilon}{2} - \frac{3 - 3\epsilon}{2} < \alpha\). Adding the edge to \(z\) would decrease stretches by \(\frac{d(b_i,a_i) + d(a_i,y)}{d(b_i,y)} - \frac{d(b_i,z) + d(z,y)}{d(b_i,y)} + \frac{d(b_i,c_i) + d(c_i,z)}{d(b_i,z)} - 1 = \frac{3.1}{2 - \epsilon} - \frac{3 - 2\epsilon}{2 - \epsilon} + \frac{3 + \epsilon}{2} - 1 < \alpha\). Since both of these moves just add edges and neither of them is an improving move individually, a combination of them cannot be an improving move either.

\textbf{Agent \(c_i\):} As noted above, the edge to \(b_i\) and \(d\) cannot be removed and the stretch to any node other than \(y\) is already \(1\) and an edge to them cannot be part of any other shorter greedy paths either. Furthermore, if \(c_i\) were to remove any of the edges to some \(t_u\) or \(f_u\), the lowest possible stretch to them even with additional edges elsewhere, would be \(\frac{d(c_i, d) + d(d,t_u)}{d(c_i, t_u)} = \frac{2.6}{1.6-\epsilon} > 1 + \alpha\) and thus including these edges would always be a better response. Thus, the only possible improving changes are dropping the edge to \(z\) and adding one to \(y\). Doing both would change stretches by \(\frac{d(c_i,z) + d(z,y)}{d(c_i,y)} - \frac{d(c_i,y) + d(y,z)}{d(c_i,z)} = \frac{3 - \epsilon}{2 + 2\epsilon} - \frac{3}{2 + \epsilon} < 0\). Adding the edge to \(y\) would decrease stretches by \(\frac{d(c_i,z) + d(z,y)}{d(c_i,y)} - 1 = \frac{3 - \epsilon}{2 + 2\epsilon} - 1 < \alpha\). Dropping the edge to \(z\) would increase stretches by \(\frac{d(c_i,t_u) + d(t_u,z)}{d(c_i,z)} - 1 + \frac{d(c_i,b_i) + d(b_i, a_i) + d(a_i,y)}{d(c_i, y)} - \frac{d(c_i,z) + d(z,y)}{d(c_i,y)} = \frac{3.2 - \epsilon}{2 + \epsilon} - 1 + \frac{4.1}{2 + 2\epsilon} - \frac{3 - \epsilon}{2 + 2\epsilon} > \alpha\).

\textbf{Agents \(t_u\) and \(f_u\):} As noted above, the edge to \(f_u\)/\(t_u\) and \(d\) cannot be removed and the stretch to any node is already \(1\) and an edge to them cannot be part of any other shorter greedy paths either. Furthermore, analogously to the edges from some \(c_i\) to some \(t_{u'}\) or \(f_{u'}\), the edges to any \(c_i\) cannot be removed either. Thus, the only possible improving change is dropping the edge to \(z\). Doing so would remove their only greedy path to \(z\) however.

Thus, there exists a NE if and only if there is a satisfying assignment for \(\mathcal{I}\) and there also exists a GE if and only if there is a satisfying assignment for \(\mathcal{I}\). Hence, \textsc{3-SAT} is reducible to our problems in polynomial time. Thus, they are NP-hard.
\end{proof}
\subsubsection*{\textbf{Approximate Equilibria}} We show the existence of very general approximate equilibria, where the approximation factor does not depend linearly on \(n\) or on the distances given by the metric but on the edge price parameter $\alpha$. Depending on the setting, this can be seen as a positive result. 

\begin{theorem}
\label{the:general-const-complete}
    Every instance of our game has a \((\alpha+1)\)-approximate NE.
\end{theorem}
\begin{proof}
Consider the complete graph on \(\mathcal{P}\). The cost of every agent is \((\alpha+1)(n-1)\) because every agent builds \(n-1\) edges and the stretch to every other node is \(1\). Since the stretch to any other agent cannot be less than \(1\), the best response of an agent has costs of at least \(n-1\). Thus, the complete graph constitutes a \((\alpha+1)\)-approximate NE.
\end{proof}

\section{Conclusion}
We give the first game-theoretic network formation model that focuses on the creation of networks where all-pairs greedy routing is enabled and where the agents optimize their connection quality, i.e., their stretches. 

We believe that this is only the first step to further models that guarantee even more favorable properties to hold in the created networks. For example, a guaranteed maximum stretch and robustness aspects like coping with edge or node failures. Another avenue for further research is to consider different edge price functions. Our goal was to first provide insights for the simplest setting, i.e., where every edge has a uniform cost of $\alpha$, that will then be useful for investigating more complex variants. And this actually works: based on the presented results we have preliminary results that hold for the version where edge costs are proportional to the edge length, as considered in~\cite{biloGeometric2019}.  Finally, it is interesting to explore other techniques for constructing approximate Nash equilibria in the Euclidean setting.

\bibliographystyle{ACM-Reference-Format}
\bibliography{gncg_references}

\end{document}